\documentclass[journal]{IEEEtran}
\IEEEoverridecommandlockouts
\usepackage{cite}
\usepackage{amsmath,amssymb,amsfonts}
\usepackage{algorithmic}
\usepackage{graphicx}
\usepackage{textcomp}
\usepackage[dvipsnames]{xcolor}

\usepackage[OT1]{fontenc} 
\usepackage{lipsum}
\usepackage{amsmath}
\usepackage{amssymb}
\usepackage{amsthm}
\usepackage{breqn}
\usepackage{thmtools, thm-restate}
\usepackage{soul}
\usepackage{cite}
\usepackage[linesnumbered,ruled,vlined]{algorithm2e}

\SetCommentSty{mycommfont}
\SetKwInput{KwInput}{Input}                
\SetKwInput{KwOutput}{Output}              

\declaretheorem{theorem}

\declaretheorem{definition}

\declaretheorem[numbered=no, name=Myerson's Lemma]{myerson}

\DeclareMathOperator{\E}{\mathbb{E}}
\DeclareMathOperator{\crit}{\text{crit}}
\let\th\relax
\DeclareMathOperator{\th}{\text{th}}

\def\BibTeX{{\rm B\kern-.05em{\sc i\kern-.025em b}\kern-.08em
    T\kern-.1667em\lower.7ex\hbox{E}\kern-.125emX}}
\begin{document}

\title{Two-Stage Auction Mechanism for Long-Term Participation in Crowdsourcing}

\author{Timothy Shin Heng Mak and Albert Y.S. Lam
    \thanks{The authors are with Fano Labs, Hong Kong. Lam is also with the University of Hong Kong (email: \{timothy, albert\}@fano.ai).}
}

\maketitle

\begin{abstract}
Crowdsourcing has become an important tool to collect data for various artificial intelligence applications and auction can be an effective way to allocate work and determine reward in a crowdsourcing platform. In this paper, we focus on the crowdsourcing of small tasks such as image labelling and voice recording where we face a number of challenges. First, workers have different limits on the amount of work they would be willing to do, and they may also misreport these limits in their bid for work. Secondly, if the auction is repeated over time, unsuccessful workers may drop out of the system, reducing competition and diversity. To tackle these issues, we first extend the results of the celebrated Myerson's optimal auction mechanism for a single-parameter bid to the case where the bid consists of the unit cost of work, the maximum amount of work one is willing to do, and the actual work completed. We show that a simple payment mechanism is sufficient to ensure a dominant strategy from the workers, and that this dominant strategy is robust to the true utility function of the workers. Secondly, we propose a novel, flexible work allocation mechanism, which allows the requester to balance between cost efficiency and equality. While cost minimization is obviously important, encouraging equality in the allocation of work increases the diversity of the workforce as well as promotes long-term participation on the crowdsourcing platform. Our main results are proved analytically and validated through simulations. 
\end{abstract}

\begin{IEEEkeywords}
crowdsourcing, mechanism design, auction, return on investment
\end{IEEEkeywords}

    \section{Introduction}
Crowdsourcing has become an established way to gather data of all kinds using the power of the crowds. Many of the most popular sites on the Internet today, such as Stack Overflow (\texttt{www.stackoverflow.com}), Wikipedia (\texttt{www.wikipedia.org}), and YouTube (\texttt{www.youtube.com}) involve some form of crowdsourcing effort from their users. Moreover, as artificial intelligence technology becomes more and more accessible, many companies are looking to crowdsourcing as a means to gather data for model training. Because of the ubiquity of smart phones, crowdsourcing has also taken on a new dimension in mobile crowdsensing (MCS) or participatory sensing, where users with smart devices contribute their location \cite{Zhou2017a} and sensory data, such as traffic data \cite{Mohan2008,Zhou2012}, pollution data \cite{Dutta2009,Sivaraman2013}, as well as market data \cite{Bulusu2008,Deng2009}. 

The rising importance of crowdsensing as an application has spawned a sizeable literature on incentive mechanisms to motivate participation \cite{Restuccia2016,Poesio2017,Capponi2019}. While monetary incentives are clearly effective, other means such as gamification \cite{vonAhn2004,Venhuizen2013} and shared intent \cite{Singh2002,JordanRaddick2013} are also being studied. Recently, a major focus has been on the use of auction and game theoretic mechanisms to motivate users to provide cost-effective and high quality input \cite{Anari2014,Luo2016,Feng2014,Zhou2017a,Zheng2017,Gao2015,Wang2017,Zhao2014,Zhou2018,Fan2016,Lee2010a,Koutsopoulos2013,Yang2016,Qin2017,Zhang2015,Singla2013,Chen2016,Chen2019,Dobakhshari2017,Gong2019,Sooksatra2019,Nie2019,Wang2019b,Xu2021,Lu2021}. Although the majority of these studies are related to mobile crowdsensing, where typical challenges include the design of \emph{location-aware} \cite{Zhou2018,Feng2014,Gao2015} and \emph{online} mechanisms \cite{Zhao2014,Fan2016,Zhou2018}, the focus of this paper is on the general crowdsourcing of repetitive tasks. In particular, our interests are in applications where a large quantity of relatively cheap manual labor is needed, such as in the collection of voice and speech data or in image classification and tagging. For this objective, two types of competitive games may be considered -- auctions and Stackelberg games \cite{Yang2012}. Auction-type games are those in which workers submit bids in the form of the minimum pay they require for the work, and the requester decides the best way to assign work and pay. Stackelberg games are those in which the requester decides on a total amount of money to be spent on the task, and workers submit bid in the form of the amount of work they would like to do. In this study, our focus is on the design of a suitable \emph{auction} mechanism. 

In particular, we are interested in a design that ensures good quality from the contributed work, for which a variety of approaches have been proposed in the literature. One approach is exemplified by \cite{Zheng2017} and \cite{Gong2019}, which require a worker to submit an indication of the quality of work beforehand, and the worker is penalized for failing to meet the declared quality. Another approach is provided by \cite{Wang2017}, \cite{Qin2017}, and \cite{Dobakhshari2017}, where the worker's quality is modelled on their past performance on the platform, and this information is taken into account in the worker's future bid or pay for work. \cite{Koutsopoulos2013} formulates quality as a constraint criteria in work allocation, which is assumed to be the sum of concave functions of the quantity of work assigned to each worker. 

Apart from work quality, ensuring continual participation from workers is another important issue in the design of a crowdsourcing platform. As pointed out by \cite{Lee2010a,Gao2015,Sooksatra2019}, workers face both direct and indirect costs when they decide to join a crowdsourcing platform. Yet, most of the research on auction mechanisms only takes into account of costs resulting directly from the performance of the work, and ignores such other costs as those involved in joining the platform in the first place. This is understandable, since workers who do not join the platform in the first place cannot reveal anything about their preference through bidding. In the minority of cases where long-term participation is considered, researchers typically model the Return on Investment ($ROI$) of a worker. 
Workers are expected to stay on the platform if their $ROI$ is above a certain worker-specific threshold. Generally, the $ROI$ can be increased by increasing a worker's probability of obtaining work. In the work of \cite{Lee2010a}, \cite{Gao2015}, and \cite{Sooksatra2019}, this is achieved by increasing the probability of obtaining work for workers whose bids are otherwise too high to be given work. For example, \cite{Lee2010a} introduced a system of virtual credits, whereby workers who fail to be allocated work are given virtual credits which would enable them to gain an advantage in subsequent auctions. \cite{Gao2015} introduced a stochastic constraint where they required the long term average of the allocation to a particular worker to be greater than a particular level. \cite{Sooksatra2019} formulated an objective in the optimization which favoured participation from different workers over assigning simply to those with the lowest costs. None of the above studies attempted to quantify or control the increase in cost due to not optimizing for the most cost-effective solution. 

In this work, we propose a two-stage auction mechanism that is particularly suitable for the crowdsourcing of small tasks, where the cost per unit of work is relatively small. In contrast to the crowdsensing literature \cite{Lee2010a,Gao2015,Sooksatra2019}, where the submission of sensing data is often automatic, in the crowdsourcing of small tasks, there is a real possibility where workers may fail (deliberately or otherwise) to complete all of the work assigned them. In this type of crowdsourcing, workers can typically afford to work on a large number of units, although there would be a limit due to time constraints and diminishing marginal returns. Moreover, it is difficult for a worker to ensure all of their work meets the quality criteria set by the requester, and therefore the requester needs to reward the worker based on the actual amount of satisfactory work. To achieve this, we first extend the optimality result of Myerson \cite{Myerson1981} and Koutsopoulos \cite{Koutsopoulos2013} to the realistic scenario where the ``bid'' of a worker consists of three parameters: the unit cost of work, the maximum work one is willing to do, and the actual work done. We prove that this two-stage auction mechanism is strategy-proof in the sense that there exists a dominant strategy for the workers to bid as well as submit work. Secondly, we propose a flexible mechanism which allows the requester to balance efficiency and equality in the allocation of work. While efficiency is essential for cost minimization, having a more equal allocation of work can increase participation probability and hence retention of workers on the auction platform. This is in turn beneficial for the requester in two ways. First, it increases competition for the work in the long term. Secondly, the quality of many tasks can benefit from an increase in the diversity of the workforce. This is particularly important when it comes to data collection for artificial intelligence applications. To the best of our knowledge, this represents the first non-crowdsensing study of long-term participation in crowdsourcing platforms. 

This rest of the paper is structured as follows. Section \ref{sec:review} reviews the standard theory on optimal reverse auction. Section \ref{sec:model} introduces and presents an analysis for our two-stage auction mechanism, including the proofs that the mechanism is strategy-proof and individual rational. Section \ref{sec:analysis} presents the analysis of our proposed allocation mechanism for balancing cost effectiveness and allocation equality. Section \ref{sec:simulations} verifies the main results with simulations. The final section gives conclusions from the present paper. 

\section{Background} \label{sec:review}
Our crowdsourcing mechanism is a form of reverse auction. In a reverse auction, a requester invites tenders from workers for work to be done. Workers indicate the preferred (monetary) compensation for the work by submitting bids and typically the worker(s) with the lowest bid(s) wins. It is customary to represent the allocation of work and payment with two vectors, $\boldsymbol{x}$ and $\boldsymbol{p}$, respectively, which are themselves functions of the vector of bids $\boldsymbol{b}$. For example, in an auction with three bids from three workers, an allocation vector of $\boldsymbol{x}=(1,2,0)$ and a payment vector of $\boldsymbol{p}=(100, 250, 0)$ mean that Worker 1 is being paid \$100 for 1 unit of work, Worker 2 \$250 for 2 units, and Worker 3 \$0 for 0 units. Supposing the cost of worker $i$ performing one unit of work is $v_i$, an auction mechanism $(\boldsymbol{x}, \boldsymbol{p})$ is known as \emph{Dominant Strategy Incentive Compatible} (DSIC) if placing the bid $b_i=v_i$ always maximizes Worker $i$'s utility $u_i$ for all $i$, where 
\begin{equation*}
    u_i = p_i - x_i v_i,
\end{equation*}
with $p_i$ being the payment and $x_i$ is the quantity of work allocated to Worker $i$. 

Denoting by $\bar{b}$ the maximum possible bid and $\boldsymbol{b}_{-i}$ the bids of all other bidders apart from Worker $i$, a fundamental Lemma to the theory of auction is the following, given first by \cite{Myerson1981}, and adapted by \cite{Koutsopoulos2013} for the reverse auction. 
\begin{myerson} 
An auction mechanism is DSIC if and only if 
\begin{enumerate}
    \item $x_i(\boldsymbol{b}) \equiv x_i(b_i, \boldsymbol{b}_{-i})$ is a non-increasing function on $b_i$ for all $i$. 
    \item $p_i(\boldsymbol{b})$ has the following form
    \begin{equation}
    p_i(\boldsymbol{b}) = b_ix_i(\boldsymbol{b}) + \int_{b_i}^{\bar{b}} x_i(s, \boldsymbol{b}_{-i}) ds + \kappa \label{eq:Myerson}
    \end{equation}
    where  $\kappa$ is a constant.
\end{enumerate}
\end{myerson}
Defining 
\begin{equation}
    \delta_i(b_i) = b_i + \frac{F(b_i)}{f(b_i)} \label{eq:delta}
\end{equation}
where $f(b_i)$ and $F(b_i)$ denote the probability density function and distribution of $b_i$ respectively, it was shown in \cite{Myerson1981,Koutsopoulos2013} that the auction minimizes cost for the requester if 
\begin{enumerate}
    \item $\boldsymbol{x}(\boldsymbol{b})$ is chosen to minimize $C = \sum_i x_i \delta_i(b_i)$, and 
    \item $\delta_i(b_i)$ is monotonic increasing in $b_i$
\end{enumerate}
This forms the basis of designing so-called \emph{optimal} auctions which are DSIC as well as cost-minimizing. Moreover, \cite{Koutsopoulos2013} showed that one can impose constraints on $\boldsymbol{x}$ in the form of $g(\boldsymbol{x}) = c$ where $g(.)$ is a concave function. The concavity of $g(.)$ preserves the monotonicity of $x_i(.)$ as a function on $b_i$ and thus the optimality of the auction. 

\section{A two-stage reverse auction mechanism} \label{sec:model}
In this section, we present the mechanism design of the two-stage reverse auction and then give our analytic results on its strategyproofness and individual rational property.

\subsection{Mechanism Design}
To adapt the standard reverse auction for crowdsourcing, we design a two-stage auction mechanism, with its workflow illustrated in Fig.~\ref{fig:flowchart}. In brief, we assume three parties in a typical round of crowdsourcing on the platform -- the \emph{\textcolor{Purple}{requester}}, the \emph{\textcolor{Green}{worker}}, and the \emph{\textcolor{Salmon}{platform}}. While the requester and the platform are assumed to act honestly, the worker is expected to submit bids in two stages, which may or may not be honest representations of his beliefs. (In order to improve the flow of the text, hereafter a worker will always be represented by ``he'' and the requester by ``she''.) In the first stage, the requester first submits a task for workers to work on. The expected quality of the work is given, together with the payment scheme, so that the worker has full information to consider his bid. The minimum payment per unit of work Worker $i$ is willing to receive and the maximum amount of work he is willing to do are denoted by $v_i$ and $x^{\max}_i$, respectively. Based on these values, he submits the bid tuple $(b_i, \hat{x}^{\max}_i)$ (the first stage bid). Note that $b_i$ and $\hat{x}^{\max}_i$ could be different from $v_i$ and $x^{\max}_i$. After all the bids are collected from all potential workers, the platform determines the expected work ($x_i$) and maximum pay ($p_i$) for each worker, where $p_i$ corresponds to the pay Worker $i$ would receive if he submits $x_i$ units of work and all $x_i$ units are satisfactory. In the second stage, the workers submit their work to the platform. The number of units of work for Worker $i$ is denoted $\hat{x}_i$ and can be different from $x_i$. The requester then assesses at most $x_i$ units of the work and determines that a proportion $\alpha_i = \tilde{x}_i / \min(x_i, \hat{x}_i)$ of the work is acceptable. The pay the worker receives is then calculated as $\tilde{p}_i=p_i\tilde{x}_i/x_i$. To aid the reading of this paper, a table of the symbols used is given in Table \ref{tab:symbols}. 
\begin{figure}[!t]
    \centering
    \includegraphics[width=\linewidth]{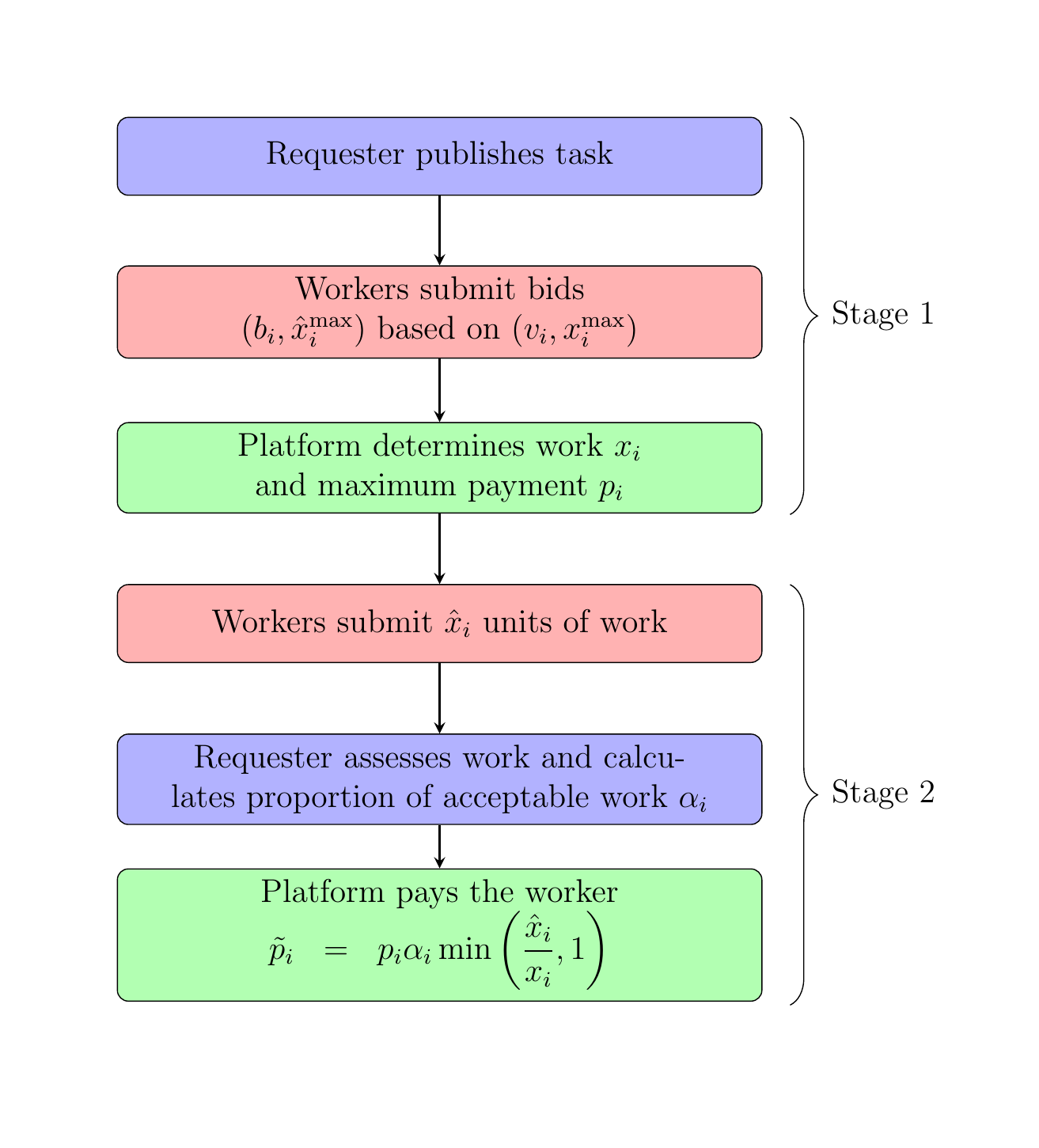}
    \caption{A two-stage reverse auction mechanism for crowdsourcing}
    \label{fig:flowchart}
\end{figure}

\begin{table}[!t]
    \caption{Description of symbols} \label{tab:symbols}
    \centering
    \begin{tabular}{lp{5cm}}
        \hline
        Parameter & Description \\
        \hline
        \multicolumn{2}{l}{\textit{Private values}} \\
        $v_i$ & Minimum unit payment worker $i$ would accept for work \\ 
        $x^{\max}_i$ & Maximum amount of work worker $i$ is willing to do \\
        \multicolumn{2}{l}{\textit{Bid}} \\
        $b_i$ & Bid corresponding to $v_i$ \\
        $\hat{x}^{\max}_i$ & Bid corresponding to $x^{\max}_i$ \\
        $\hat{x}_i$ & The actual amount of work done by worker $i$ \\
        \multicolumn{2}{l}{\textit{Platform calculated values}} \\
        $\boldsymbol{x}=(x_1, \ldots, x_n)$ & Amount of work to workers \\
        $\boldsymbol{p}=(p_1, \ldots, p_n)$ & Maximum promised payment to workers \\
        $\alpha_i$ & Proportion of acceptable work submitted by worker $i$ \\
        $\tilde{x}_i$ & Quantity of acceptable work by worker $i$ \\
        $\tilde{p}_i$ & Payment to worker $i$ \\
        $U_i$ & Utility of worker $i$ \\
        \multicolumn{2}{l}{\textit{Hyperparameters}} \\
        $C$ & Total expected cost for requester \\
        $c$ & Total amount of work requested \\
        $\beta_i$ & $\E \alpha_i$ \\
        $f(.), F(.)$ & Density and distribution function of bids \\
        $\delta_i = \delta_i(b_i)$ & ``Virtual welfare'' of worker $i$ see equation \eqref{eq:delta} \\
        $\bar{b}$ & Maximum possible bid \\
        $k$ & Parameter for balancing equality against efficiency in work allocation \\
        $n$ & Number of workers in the auction \\
        \hline
    \end{tabular}
\end{table}

To formalize ideas, we write the utility function from the worker's perspective as
\begin{align}
U_i&(b_i, \hat{x}^{\max}_i, \hat{x}_i) = 
\begin{cases}
\E_{\alpha_i} \tilde{p}_i - \hat{x}_i v_i & \text{if } \hat{x}_i \leq x^{\max}_i \\
0 & \text{otherwise}
\end{cases} \nonumber \\
&= I(\hat{x}_i \leq x^{\max}_i) \nonumber\\
& \qquad \left[ \beta_i \min \left( \frac{\hat{x}_i}{x_i(\hat{x}^{\max}_i, b_i)}, 1 \right) p_i(\hat{x}^{\max}_i, b_i) - \hat{x}_iv_i \right]. \label{eq:utility00}
\end{align} 
where $\beta_i = \E \alpha_i$, $I(.)$ denotes the indicator function, and
\begin{equation}
p_i(\hat{x}^{\max}_i, b_i) = b_ix_i(\hat{x}^{\max}_i, b_i) + \int_{b_i}^{\bar{b}} x_i(\hat{x}^{\max}_i, s) ds  \label{eq:MyersonPay}
\end{equation}
is the payment rule for maximum pay, which is the same as \eqref{eq:Myerson} except that $x_i$ and $p_i$ are now functions of both $\boldsymbol{b}$ and $\hat{\boldsymbol{x}}^{\max}$ and we set $\kappa=0$. In \eqref{eq:MyersonPay} and in the rest of the paper, to simplify notation, the dependence of $x_i$ and $p_i$ on some of the parameters, especially $\boldsymbol{b}_{-i}$ and $\hat{\boldsymbol{x}}^{\max}_{-i}$, are often left implicit. 

In \eqref{eq:utility00}, we assume that the utility drops to zero when the work $\hat{x}_i$ exceeds $x^{\max}_i$. We note that this is for the simplicity of presentation, and that the form of $U_i$ can in fact be much more flexible. In Appendix \ref{sec:relax}, we show that the form of $U_i$ can be 
\begin{equation}
    U_i(b_i, \hat{x}^{\max}_i, \hat{x}_i) = \omega_i(\hat{x}_i)(\E_{\alpha_i} \tilde{p}_i - \hat{x}_i v_i) \label{eq:trueutility}
\end{equation}
where any function of $\omega_i(\hat{x}_i)$ satisfying 
\begin{equation}
    \omega_i(\hat{x}_i) 
    \begin{cases}
        = 1 & \text{if } \hat{x}_i \leq x^{\max}_i \\
        < \frac{x^{\max}_i}{\hat{x}_i} & \text{if } \hat{x}_i > x^{\max}_i
    \end{cases} \label{eq:omega}
\end{equation}
will give us the same conclusions of this paper. 

\subsection{Strategyproofness and individual rationality}
We now show that if the allocation function $x_i(b_i, \hat{x}_i)$ is non-increasing on $b_i$, then the two-stage reverse auction mechanism with utility function \eqref{eq:utility00} and maximum payment function \eqref{eq:MyersonPay} is strategyproof in that there exists a dominant strategy for all workers to bid and to complete work (Theorem \ref{theorem1}). Furthermore, if a worker plays by the dominant strategy, then the mechanism is individual rational for him (Theorem \ref{theorem2}). 
To prove these two results, we first introduce the following three Lemmas. 
\begin{restatable}{Lemma}{Lemmaone} \label{Lemma1}
    Suppose a worker completes all the work allocated to him, i.e. $\hat{x}_i = x_i \leq x^{\max}_i$. If $x_i(b_i)$ is a non-increasing function of $b_i$, a dominant strategy is to bid $b_i = v_i / \beta_i$, where $\beta_i = \E\alpha_i$. 
\end{restatable}
\begin{proof}
    See Appendix \ref{sec:Lemma1}.
\end{proof}
\begin{restatable}{Lemma}{Lemmatwo}  \label{Lemma2}
    Given $b_i = v_i/\beta_i$. For any $x_i$, a dominant strategy is to play $\hat{x}_i=\min(x_i, x^{\max}_i)$. 
\end{restatable}
\begin{proof}
    See Appendix \ref{sec:Lemma2}.
\end{proof}

\begin{restatable}{Lemma}{Lemmathree}  \label{Lemma3}
    Given $(b_i, \hat{x}_i) = (v_i/\beta_i, \min(x_i, x^{\max}_i))$, $\hat{x}^{\max}_i = x^{\max}_i$ is a dominant strategy. 
\end{restatable}
\begin{proof}
    See Appendix \ref{sec:Lemma3}.
\end{proof}

These Lemmas together imply the following:

\begin{restatable}{theorem}{theo} \label{theorem1}
    Supposing $x_i(b_i, \hat{x}^{\max}_i) \equiv x_i(b_i, \boldsymbol{b}_{-i}, \hat{x}^{\max}_i, \hat{\boldsymbol{x}}^{\max}_{-i})$ to be a non-increasing allocation function of $b_i$ for all $\hat{\boldsymbol{x}}^{\max}$ and $\boldsymbol{b}_{-i}$, $(b_i, \hat{x}^{\max}_i, \hat{x}_i) = (v_i/\beta_i, x^{\max}_i, x_i)$ is the only dominant strategy given the utility function \eqref{eq:utility00}. 
\end{restatable}
\begin{proof}
    See Appendix \ref{sec:theorem1}.
\end{proof}
Theorem \ref{theorem1} shows that providing the allocation $x_i(b_i, \hat{x}^{\max}_i)$ is non-increasing on $b_i$, the two-stage auction mechanism is basically truthful, except that the worker needs to bid $v_i/\beta_i$ instead of $v_i$ to account for the fact that not all of his work would be satisfactory. In this work we assume that $\beta_i$, the expected proportion of satisfactory work, is fixed and known to the worker. However, in practice, the workers may be able to control $v_i$ and $\beta_i$ to some extent. For example, a worker putting in more effort will generally lead to a higher $v_i$ as well as higher $\beta_i$. The worker's strategy would be to find $v_i$ and $\beta_i$ that minimize the ratio $v_i/\beta_i$, since the proof of the following theorem implies that utility is a decreasing function of $v_i/\beta_i$. 

\begin{theorem} \label{theorem2}
    The mechanism is individual rational, i.e. $U_i(v_i/\beta_i, \hat{x}^{\max}_i, x_i) \geq 0$, if the worker plays by the dominant strategy on $b_i$ and $\hat{x}_i$. 
    \begin{proof}
        \begin{align}
        U_i&(v_i/\beta_i, \hat{x}^{\max}_i, x_i) \nonumber \\
        &=\beta_i p_i(\hat{x}^{\max}_i, v_i/\beta_i) - x_i(\hat{x}^{\max}_i, v_i/\beta_i)v_i \nonumber \\
        &=\beta_i \left[ \frac{v_i}{\beta_i} x_i(\hat{x}^{\max}_i, v_i/\beta_i) + \int_{v_i/\beta_i}^{\bar{b}_i} x_i(\hat{x}^{\max}_i, s) ds \right] - \nonumber \\
        & \quad \ x_i(\hat{x}^{\max}_i, v_i/\beta_i)v_i \nonumber \\ 
        &=\beta_i \int_{v_i/\beta_i}^{\bar{b}_i} x_i(\hat{x}^{\max}_i, s) ds \geq 0 \label{eq:theorem2}
        \end{align}
    \end{proof}
\end{theorem}
Theorem \ref{theorem2} ensures that the workers cannot lose money or have negative utility simply by participating in the crowdsourcing platform. 

\section{An allocation mechanism to increase worker participation} \label{sec:analysis}
Our second main contribution in this paper is the proposal of the following allocation mechanism for balancing \emph{cost efficiency} against \emph{allocation equality}. While cost efficiency is obviously important to the requester, we show that increasing equality in allocation can increase long term participation and diversity from the workers. 

Recall that in the standard optimal auction, allocation is implemented by minimizing the linear function $\sum_i \delta_ix_i$, which is shown by \cite{Myerson1981} and \cite{Koutsopoulos2013} to equal the total expected cost. Here, our proposed allocation rule is the minimization of the following quadratic program. 
\begin{align}
\text{minimize} \quad & \sum_i \delta_i^k x_i^2 \nonumber \\
\text{subject to} \quad & 0 \leq x_i \leq \hat{x}^{\max}_i, \quad \forall i \nonumber \\
& \sum_i x_i = c \label{eq:problem}
\end{align}
where $c$ is the amount of work required by the requester and can be set in relation to the number of workers available ($n$). Evidently, a larger $c$ leads to more work being done, but at the same time reduces competition leading to less value for money. 

Note that not all of the inequality constraints in \eqref{eq:problem} are necessarily \emph{tight}, where we define a \emph{tight} constraint to be one for which the solution to \eqref{eq:problem} would have been different if the constraint was not present. In Appendix \ref{sec:Lemma32}, we show that so long as $\delta_i > 0$ for all $i$, none of the constraints $x_i \geq 0$ are tight, and can be safely ignored if we assume all of our bids $b_i$ are positive. Because there are precisely $n$ remaining inequality constraints for $n$ workers, we denote the set of tight constraints corresponding to $x_i \leq \hat{x}^{\max}_i$ by $\mathcal{A}^* \subseteq \{1,\ldots,n \}$.\footnote{Since there is one constraint for each Worker $i$, with a slight abuse of notation, we index a constraint in $\mathcal{A}^*$ by $i$ in this article.} Moreover, since our optimization Problem \eqref{eq:problem} is a strictly convex quadratic program, we can also give a mathematically precise definition of a tight constraint as follows. 
\begin{definition}
    A constraint $x_i \leq \hat{x}^{\max}_i$ is \emph{tight}, i.e. $i \in \mathcal{A}^*$, if and only if its Lagrange multiplier is non-zero in the solution of the Lagrangian dual of \eqref{eq:problem}. 
\end{definition}


With this definition, the following Lemma gives the standard Karush-Kuhn-Tucker (KKT) solution to \eqref{eq:problem}. 
\begin{restatable}{Lemma}{Lemmathreetwo}  \label{Lemma32}
	Denoting by ${\mathcal{A}^*}^- = \{1,\ldots, n\} \setminus \mathcal{A}^*$ the complement of $\mathcal{A}^*$, the solution set of \eqref{eq:problem} is given by
	\begin{equation}
	x_i = \begin{cases}
	\hat{x}^{\max}_i & \text{if } i \in \mathcal{A}^* \\
	c'(\mathcal{A}^*) \delta_i^{-k}/\sum_{j \in {\mathcal{A}^*}^-} \delta_j^{-k} & \text{otherwise} 
	\end{cases}
	\end{equation}
	where $c'(\mathcal{A}^*) = c - \sum_{j \in \mathcal{A}^*} \hat{x}^{\max}_i$. 
\end{restatable}
\begin{proof}
	See Appendix \ref{sec:Lemma32}.
\end{proof}

Note that Lemma \ref{Lemma32} implies that the solution $\boldsymbol{x}$ and the tight constraints $\mathcal{A}^*$ are mutually dependent on one another. Thus, a heuristical approach to solving for $\boldsymbol{x}$ and $\mathcal{A}^*$, as formalized in Algorithm \ref{algorithm}, is to update them iteratively. In this algorithm, we initialize $\mathcal{A}^{(0)}$ to the empty set (line 1). In the first step ($m=1$), we solve for $\boldsymbol{x}$ given $\mathcal{A}^{(m-1)}$ (line 3). All $x_i$s in the solution which violate the constraint $x_i \leq \hat{x}^{\max}_i$ are added to $\mathcal{A}^{(m)}$ (line 5). This process is iterated until no more constraints are violated, whence we have our estimates for $\boldsymbol{x}$ and $\mathcal{A}^*$. Note that Algorithm \ref{algorithm} consists only of simple arithmatic operations and its complexity is $\leq \mathcal{O}(n^2)$. 

\begin{algorithm}[!t]
    \DontPrintSemicolon
    \SetKwRepeat{Do}{do}{while}
    \SetKwRepeat{Repeat}{repeat}{end}
    \SetKw{ForAll}{for all}
    \SetKw{ForAny}{for any}
    
    \KwInput{$\boldsymbol{\delta}, \hat{\boldsymbol{x}}^{\max}, c$}
    \KwOutput{$\boldsymbol{x}, \mathcal{A}$}
    $\mathcal{A}^{(0)} := \emptyset$; $m := 1$
    
    \Repeat{}{
        For all $i$, $x^{(m)}_i = 
        \begin{cases}
        \hat{x}^{\max}_i & \text{if } i \in \mathcal{A}^{(m-1)} \\
        c'(\mathcal{A}^{(m-1)}) \delta_i^{-k}/\sum_{j \in {\mathcal{A}^{(m-1)}}^-} \delta_j^{-k} & \text{otherwise}
        \end{cases}$
        
        \Indp where $c'(\mathcal{A}^{(m-1)}) = c - \sum_{j \in {\mathcal{A}^{(m-1)}}^-} \hat{x}^{\max}_j$. 
        
        \Indm
        \If{$x^{(m)}_i > \hat{x}^{\max}_i$ \ForAny $i$}{
            $\mathcal{A}^{(m)} = \mathcal{A}^{(m-1)} \cup \{i: x^{(m)}_i > \hat{x}^{\max}_i \}$
        }
        \Else {
            \Return{$\boldsymbol{x}=\boldsymbol{x}^{(m)}, \mathcal{A}=\mathcal{A}^{(m-1)}$}
        }
        
        $m=m+1$
        
    }
    
    \caption{Algorithm for allocating work $\boldsymbol{x}$}
    \label{algorithm}
\end{algorithm}

It can be shown that Algorithm \ref{algorithm} in fact converges to the optimum of \eqref{eq:problem}. However, before that, it is useful to further examine the properties of Problem \eqref{eq:problem}. First, note that the KKT conditions (Appendix \ref{sec:Lemma32}) imply that at the minimum: 
\begin{equation*}
\mu(\mathcal{A}^*) =  c'(\mathcal{A}^*)\left(\sum_{j \in {\mathcal{A}^*}^-} \delta_j^{-k} \right)^{-1} 
\end{equation*}
and
\begin{equation}
\lambda_i(\mathcal{A}^*) = 
\begin{cases}
\tilde{\lambda}_i(\mathcal{A}^*)  = \mu(\mathcal{A}^*) - \delta_i^{k} \hat{x}^{\max}_i > 0 & \text{iff } i \in \mathcal{A}^* \\
0 & \text{otherwise}, 
\end{cases} \label{eq:kkt2}
\end{equation}
where $\lambda_i \geq 0$ is the Lagrange multiplier for the constraint $x_i \leq \hat{x}^{\max}_i$ and $\mu$ is the (negated) Lagrange multiplier for the constraint $\sum_i x_i = c$. Note that by definition, $i \in \mathcal{A}^*$ if and only if $\tilde{\lambda}_i(\mathcal{A}^*)  = \mu(\mathcal{A}^*) - \delta_i^{k} \hat{x}^{\max}_i > 0$. Thus, denoting by $\gamma_i = \delta_i^{k} \hat{x}^{\max}_i$, if $i \in \mathcal{A}^*$, then $j \in \mathcal{A}^*$ for all $\{j : \gamma_j < \gamma_i \}$. Hence, to simplify notation, we can assume without loss of generality that $\gamma_1 \leq \gamma_2 \leq \ldots \leq \gamma_n$, and $\mathcal{A}^* \in \mathbb{H} = \{\emptyset, \{1\}, \{1,2\}, \ldots, \{1,2,\ldots,n \} \}$. To further simplify notation, we denote by $\mathcal{A}_0 = \emptyset, \mathcal{A}_1 = \{1\}, \mathcal{A}_2 = \{1, 2\}$, and so on. We now state the following two Lemmas, whose proofs are given in Appendices \ref{sec:Lemma3a} and \ref{sec:Lemma_convergence} respectively. 

\begin{restatable}{Lemma}{propositionone} \label{Lemma3a}
	$\gamma_i < \mu(\mathcal{A}_r) \implies \gamma_i < \mu(\mathcal{A}_q)$ for all $0 \leq r < q < i \leq n$. 
	\label{propositionone}
\end{restatable}

\begin{restatable}{Lemma}{Lemmaconv}  \label{Lemma_convergence}
    Algorithm \ref{algorithm} converges to the solution of \eqref{eq:problem}. 
\end{restatable}

Lemma \ref{Lemma_convergence} in turn implies the following Theorem, whose proof is given in Appendix \ref{sec:theorem3}.  
\begin{restatable}{theorem}{theoremthree}\label{theorem3}
    Given the quadratic programming problem of \eqref{eq:problem}, $x_i(b_i, \hat{x}^{\max}_i) \equiv x_i(b_i, \boldsymbol{b}_{-i}, \hat{x}^{\max}_i, \hat{\boldsymbol{x}}^{\max}_{-i})$ is a continuous, non-increasing allocation function of $b_i$ for all $\hat{\boldsymbol{x}}^{\max}$ and $\boldsymbol{b}_{-i}$. 
\end{restatable}
Thus by Theorem \ref{theorem1} and Theorem \ref{theorem3}, $(b_i, \hat{x}^{\max}_i, \hat{x}_i)=(v_i, x^{\max}_i, x_i)$ is our dominant strategy for our allocation mechanism \eqref{eq:problem}. 

Note that our mechanism is not cost minimizing in its general form. However, the parameter $k$ in \eqref{eq:problem} controls the degree of \emph{efficiency} vs.\ \emph{equality}. While it is obvious that setting $k=0$ results in the equal allocation of work, we show that setting $k=\infty$ is equivalent to the cost minimizing approach of \cite{Koutsopoulos2013}. 
\begin{restatable}{theorem}{auctionmin} \label{theorem4}
    In our allocation mechanism, when $k=\infty$, the auction minimizes total expected cost. 
\end{restatable}
\begin{proof}
	See appendix \ref{sec:theorem4}.
\end{proof} 

Now, denote by $\pi^{(k, \boldsymbol{b})}(b_i) = x^{(k)}_i(\boldsymbol{b})/c$ where $(k)$ indicates the dependence of $x_i(.)$ on $k$, and $\pi^{(k, \boldsymbol{b})}(.)$ can be interpreted as a probability mass function for the distribution of work among the workers. The following Lemma implies that as $k$ increases, workers with lower bids will get increasingly more share of the work while workers with higher bids will get less. 
\begin{restatable}{Lemma}{Lemmasc} \label{LemmaSC}
	When comparing the distributions $\pi^{(k, \boldsymbol{b})}(b_i)$ and $\pi^{(k+1, \boldsymbol{b})}(b_i)$, there always exists $t$ with $\pi^{(k, \boldsymbol{b})}(b_i) \leq \pi^{(k+1, \boldsymbol{b})}(b_i)$ for all $b_i < t$ and $\pi^{(k, \boldsymbol{b})}(b_i) \geq \pi^{(k+1, \boldsymbol{b})}(b_i)$ for all $b_i > t$. 
\end{restatable}
\begin{proof}
	See appendix \ref{sec:LemmaSC}. 
\end{proof}

Moreover, the following theorem shows that as $k$ increases, the total expected cost decreases. 
\begin{restatable}{theorem}{monodec} \label{theorem5}
    The total expected cost $C$ is a non-increasing function of $k$. 
\end{restatable}
\begin{proof}
	See appendix \ref{sec:theorem5}
\end{proof} 

Lemma \ref{LemmaSC} and Theorem \ref{theorem5} imply that the requester can control the degree of extra spend for a more equal distribution of work. Here we show how this can motivate long term participation on the crowdsourcing platform by examining the $ROI$ for the participants. While the general definition of $ROI$ is given by 
\begin{equation*}
ROI = \frac{\text{Net profit}}{\text{Direct cost} + \text{Indirect costs}}, 
\end{equation*}
here we operationally define the ROI for Worker $i$ as
\begin{equation}
    ROI_i = \frac{\E (\tilde{p}_i - x_i v_i) - \gamma_i}{\E x_i v_i + \gamma_i} = \frac{\E U_i - \gamma_i}{\E x_i v_i + \gamma_i}  \label{eq:roi1}
\end{equation} 
where the expectation is taken with respect to the private values of all other workers, i.e. $(\boldsymbol{b}_{-i}, \hat{\boldsymbol{x}}^{\max}_{-i}, \hat{\boldsymbol{x}}_{-i})$, while holding the worker's own value constant. $\E x_i v_i$ then represents the expected direct cost in participation in the crowdsourcing platform, and $\gamma_i$ the indirect cost, and we assume the dominant strategy is followed such that $\hat{x}_i=x_i, \forall i$. We note that this definition is in the same spirit as \cite{Gao2015}. Generally, $\gamma_i$ is assumed to be $0$, and the individual rational property of the mechanism (Theorem \ref{theorem2}) implies $ROI_i \geq 0$. When $\gamma_i > 0$, it is possible for $ROI_i$ to be negative, resulting in a net loss for the worker. In this paper, we assume that Worker $i$ will drop out of the platform if their $ROI_i$ is less than 0 (since this represents long term loss), in other words, if $\E U_i < \gamma_i$. 
From Theorem \ref{theorem3}, $x_i$ is a continuous function of $b_i$. From Lemma \ref{LemmaSC}, we know that $x^{(k)}_i(b_i=\bar{b}) \geq x^{(k+1)}_i(b_i = \bar{b})$. Therefore, from \eqref{eq:theorem2}, we can expect 
\begin{equation*}
U^{(k)}_i = \beta_i \int_{b_i}^{\bar{b}} x^{(k)}_i(s) ds \geq U^{(k+1)} = \beta_i \int_{b_i}^{\bar{b}} x^{(k+1)}_i(s) ds
\end{equation*} 
for workers whose $b_i$ are close to $\bar{b}$, the upper end of its domain. In other words, we can expect as $k$ decreases, the work ($x_i$), the utility ($U_i$) as well as $ROI_i$ of those with higher $b_i(=v_i/\beta_i)$ to go up, and hence their likelihood staying in the platform. Although those with lower values of $b_i$ could possibly have decreased $U_i$ with decreasing $k$, they will still have higher utility than those with higher bids (assuming equal $\hat{x}^{\max}_i$ and $\gamma_i$) and thus they will not be leaving the platform. Hence, overall, there will be increased participation. We will illustrate this in the simulation study in the following section.

\section{Simulations} \label{sec:simulations}
In this section, we validate our analytic results introduced in the previous section through simulations.  Source codes for the simulations may be found in the first author's github page\footnote{\texttt{https://github.com/tshmak/TwoStageAuctionSimulation}}. Simulations are run on a MacBook Pro (Intel i7 CPU, 32GB RAM) using R version 3.5.1 using a single CPU.  
For the results that follow, we assume the worker's bids $b_i$ come from a truncated log-Normal distribution 
\begin{equation*}
b_i \overset{\text{iid}}{\sim} \text{TruncatedLogNormal}(\mu=0, \sigma=0.3), b_i \in (0, \bar{b}=2.01)
\end{equation*} 
where the truncation point ($\bar{b}$) corresponds roughly to the $99^{\th}$ percentile of the un-truncated distribution. The $5^{\th}, 25^{\th}, 50^{\th}, 75^{\th}$, and $95^{\th}$ percentiles of this distribution are 0.61, 0.81, 1.00, 1.22, and 1.60 respectively. The scale of $b_i$ does not matter and hence we arbitrarily set the distribution to be roughly centered at 1. Under this distribution, the largest bid is typically around 3 to 4 times the smallest bid, which we believe reflects the pattern in real-life bids on such crowdsourcing platform. It can be verified that this log-Normal distribution is \emph{regular}, i.e. $\delta(b_i)$ is a monotonic increasing function of $b_i$. Furthermore, we assume $\hat{x}^{\max}_i = x^{\max}_i$ and $\beta_i$ are randomly and independently drawn from 
\begin{align*}
x^{\max}_i  &\sim 100\times \text{logNormal}(\mu=0, \sigma=0.3) \\
\beta_i &\sim \text{Uniform}(0.9, 1)
\end{align*}
Thus, $x^{\max}_i$ has a similar distribution to $b_i$, except scaled by 100. We set $c$ in the constraint on $\sum_i x_i$ to 
\begin{equation*}
c = \sum_i x_i = 100n\rho
\end{equation*}
where $\rho \in \{0.1, 0.5\}$, such that the total allocated work roughly equals either 0.1 or 0.5 of the total workers' capacity. Evidently, $\rho=0.1$ will result in a more competitive auction than $\rho=0.5$, and is expected to yield a higher efficiency, but lower $ROI$ for the workers. As our theoretical results show that $k$ has a simple monotonic relationship with the allocation of work (Lemma \ref{LemmaSC}) and expected costs (Theorem \ref{theorem5})), we set $k \in \{0,1,2,4,8, \infty \}$. The similarity of the simulation results for $k=8$ and $k=\infty$ suggest that this choice effectively covers the entire domain of $k$.  For $n$, we set $n \in \{10, 100, 1000\}$, which represents what we believe are typical numbers of workers encountered in real crowdsourcing platforms. We repeat the simulations 100 times for each combination of $k$, $n$, and $\rho$. 

\subsection{Verification of Theorem \ref{theorem1}}
We first verify that, from Theorem \ref{theorem1}, $(b_i, \hat{x}^{\max}_i, \hat{x}_i)=(v_i/\beta_i, x^{\max}_i, x_i)$ is a dominant strategy. To show this, we randomly generate vectors of $\boldsymbol{b}$, $\boldsymbol{x}^{\max}$, and $\boldsymbol{\beta}$ as above. Without loss of generality, we note the utility for one of the workers, say Worker 1, given $(b_1, \hat{x}^{\max}_1, \hat{x}_1)=(v_1/\beta_1, x^{\max}_1, x_1)$. In order to increase the probability of $\hat{x}^{\max}_i$ being limiting (i.e., $=x_i$, because otherwise $\hat{x}^{\max}_1$ has no effect on the allocation $x_i$), we restrict our attention to samples with $F(b_1)<\rho$. We use the \texttt{bobyqa} method \cite{Powell2009} to search the neighbourhood for a better set of parameters. We let the utility function of Worker 1 take the form of \eqref{eq:trueutility}, with 
\begin{equation*}
\omega(\hat{x}_1) = 
\begin{cases}
    1 & \text{if } \hat{x}_1 \leq x^{\max}_1 \\
    1 + \frac{\hat{x}_1 - x^{\max}_1}{x^{\max}_1}s & \text{if } \hat{x}_1 > x^{\max}_1
\end{cases}
\end{equation*}
being the downweighting function which reflects the diminishing value of return as $\hat{x}_i > x^{\max}_i$. Here, $s < 0$ is a slope parameter which determines the rate at which $\omega$ decreases from 1. Our result (Appendix \ref{sec:relax}) implies that  if $s \leq -1$, $\omega(\hat{x}_1)$ satisfies the form of \eqref{eq:omega}, and $(b_1, \hat{x}^{\max}_1, \hat{x}_1)=(v_1/\beta_1, x^{\max}_1, x_1)$ will be a dominant strategy. We verify (using our search algorithm) that no better parameters than the nominally dominant are found when $s< -1$. For $-1 < s < 0$, we tabulate the average departure of the optimal parameters, denoted by $\theta^*_1 = (b^*_1, {{}\hat{x}^{\max}}^*_1, \hat{x}^*_i)$, from the nominally dominant $\theta_1 = (v_1/\beta_1, x^{\max}_1, x_1)$, for $s \in \{-0.25, -0.5\}$ in Table \ref{tab:verify}. In a large proportion of cases, $(b^*_1, {{}\hat{x}^{\max}}^*_1, \hat{x}^*_i)=(v_1/\beta_1, x^{\max}_1, x_1)$, i.e. the nominally dominant strategy continues to be the best. The proportion of cases in which ${{}\hat{x}^{\max}}^*_1 \neq x^{\max}_1$ increases with $k$, as does the relative difference of ${{}\hat{x}^{\max}}^*_1$ and $x^{\max}_1$ when ${{}\hat{x}^{\max}}^*_1 \neq x^{\max}_1$. To a lesser extent, this is observed for $b^*_1$ as well. 

\begin{table}[!t]
    \centering
    \caption{Departure from the nominally dominant strategy when the utility function of worker 1 does not satisfy \eqref{eq:omega}} \label{tab:verify}
    \begin{tabular}{cccc}
        \hline
        \multicolumn{2}{c}{
        $\begin{array}{c}
        \%\{\theta^*_1 = \theta_1 \} \\
        \left( \begin{array}{c}
        \text{Mean} \\ \text{rel. diff.}\dagger
        \end{array} | \theta^*_1 \neq \theta_1 \right) \\
        \end{array}
        $} & $s=-0.5$ & $s=-0.25$ \vspace{1mm} \\
        \hline
$k=0$ & $\theta_1 = \begin{array}{c} v_1/\beta_1 \\ x^{\max}_1 \\ x_1 \end{array}$ & $\begin{array}{c} 100\%\ (NA) \\99\%\ (0.15) \\100\% (NA) \\ \end{array}$ & $\begin{array}{c} 100\%\ (NA) \\99\%\ (0.13) \\100\% (NA) \\ \end{array}$ \vspace{1mm} \\
$k=1$ & $\theta_1 = \begin{array}{c} v_1/\beta_1 \\ x^{\max}_1 \\ x_1 \end{array}$ & $\begin{array}{c} 100\%\ (NA) \\100\%\ (0.097) \\100\% (NA) \\ \end{array}$ & $\begin{array}{c} 100\%\ (NA) \\97\%\ (0.11) \\100\% (NA) \\ \end{array}$ \vspace{1mm} \\
$k=2$ & $\theta_1 = \begin{array}{c} v_1/\beta_1 \\ x^{\max}_1 \\ x_1 \end{array}$ & $\begin{array}{c} 100\%\ (0.25) \\97\%\ (0.064) \\100\% (NA) \\ \end{array}$ & $\begin{array}{c} 100\%\ (0.25) \\88\%\ (0.23) \\100\% (NA) \\ \end{array}$ \vspace{1mm} \\
$k=4$ & $\theta_1 = \begin{array}{c} v_1/\beta_1 \\ x^{\max}_1 \\ x_1 \end{array}$ & $\begin{array}{c} 100\%\ (NA) \\85\%\ (0.13) \\100\% (NA) \\ \end{array}$ & $\begin{array}{c} 100\%\ (0.089) \\73\%\ (0.45) \\100\% (NA) \\ \end{array}$ \vspace{1mm} \\
$k=8$ & $\theta_1 = \begin{array}{c} v_1/\beta_1 \\ x^{\max}_1 \\ x_1 \end{array}$ & $\begin{array}{c} 100\%\ (0.18) \\67\%\ (0.2) \\100\% (NA) \\ \end{array}$ & $\begin{array}{c} 100\%\ (0.11) \\53\%\ (0.66) \\100\% (NA) \\ \end{array}$ \vspace{1mm} \\
$k=\infty$ & $\theta_1 = \begin{array}{c} v_1/\beta_1 \\ x^{\max}_1 \\ x_1 \end{array}$ & $\begin{array}{c} 99\%\ (0.17) \\23\%\ (0.43) \\100\% (NA) \\ \end{array}$ & $\begin{array}{c} 98\%\ (0.35) \\21\%\ (1.2) \\100\% (NA) \\ \end{array}$ \vspace{1mm} \\
\hline
\multicolumn{4}{l}{$\dagger$ Mean rel. diff. = $\text{Mean}\left( \frac{| \theta^*_1 - \theta_1 |}{\theta_1}\right)$}
    \end{tabular} 
\end{table}

\subsection{Balancing cost and workers' $ROI$}
In Section \ref{sec:analysis}, we show how the parameter $k$ in \eqref{eq:problem} affects the total cost on the requester's side and the $ROI$ on the worker's side. In this simulation, in order to estimate the $ROI$, we set the bid of the first worker to $b_1 = F^{-1}(a), x^{\max}_1=100$, where $a \in \{0.1, 0.2, \ldots, 0.9\}$. In other words, his bid is set to the $10^{\th}, \ldots, 90^{\th}$ percentile of the distribution of bids. We also assume this worker's $\beta_1=0.95$. $ROI$ is calculated as 
\begin{equation}
ROI(v_1) = \frac{\sum_j^{100} \tilde{p}_{1j}/100}{\sum_j^{100} x_{1j} v_1/100 + \gamma_1} - 1 \label{eq:ROI}
\end{equation}
where $\tilde{p}_{1j}$ and $x_{1j}$ are the payment and workload in the $j^{\th}$ repeat for the first worker, respectively, and $\gamma_1 \in \{0,1,2,3,4,5\}$ is the indirect cost. Thus, the indirect cost is roughly 0 to 5\% of the cost of performing $x^{\max}_1=100$ units of work. 

Fig.~\ref{fig:ROI} shows the relationship between the $ROI(v_1)$ and $v_1$. We use monotonic smoothed estimates of $ROI$ to improve accuracy, where the smoothing is done using the default settings in the \texttt{scam} package in \texttt{R} \cite{Pya2014}.  
\begin{figure}[!t]
    \centering
    \centerline{\includegraphics[width=1\linewidth]{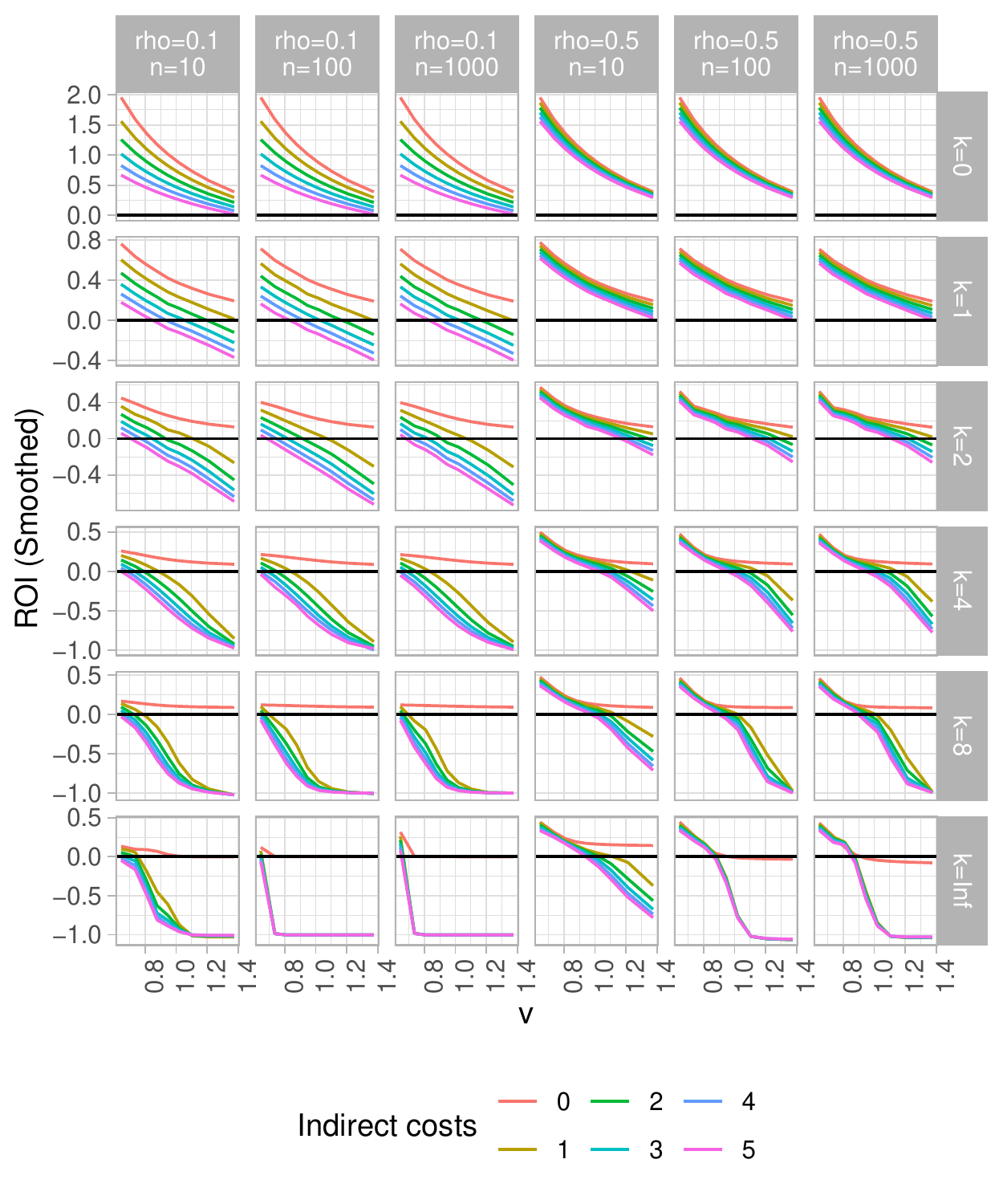}}
    \caption{Relationship between $ROI$ and unit cost of work $v_1$}
    \label{fig:ROI}
\end{figure}
In the figure, we see that $ROI$ is always positive when $\gamma_1=0$. This corresponds to the individual rationality property of the mechanism. $ROI$ also decreases with $v_i$, according to our predictions in Section \ref{sec:analysis}. $ROI$ increases with decreasing $k$, corresponding to increasing equality in the allocation of work. It also increases with $\rho$, because of the increased probability of obtaining work. Notably, the impact of $n$ on $ROI$ is relatively small. 

In Fig.~\ref{fig:participation}, we examine the impact of $k$ on long-term participation. We assume workers with $ROI(v_i) \geq 0$ will continue to participate, whereas those with $ROI(v_i) < 0$ will drop out. The proportion of $ROI(v_i) \geq 0$ is estimated as the proportion of $v_i$ to the left of the point at which the $ROI$ curve in Fig.~\ref{fig:ROI} crosses 0. 
\begin{figure}[!t]
    \centering
    \includegraphics[width=1\linewidth]{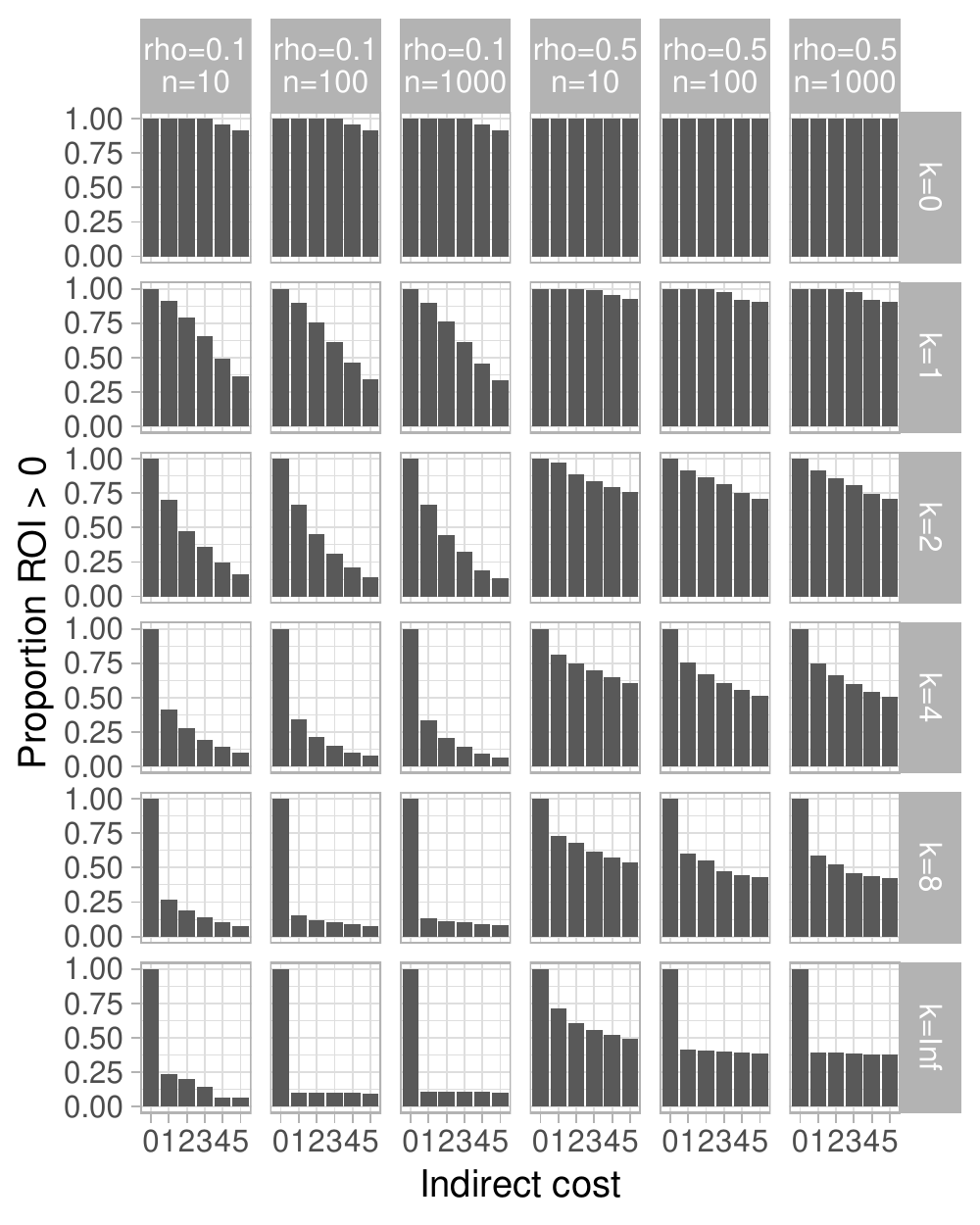}
    \caption{Relationship between $k$ and probability of $ROI$ greater than 0}
    \label{fig:participation}
\end{figure} 
Evidently, participation increases as $k$ decreases. It also increases with $\rho$. For example, when $k=2, n=1000, \rho=0.5, \gamma_1=3$, the proportion of workers with $ROI$ greater than 0 is around 0.8. This means that around 20\% of the workforce is expected to drop out from the platform. The platform can only be sustained at this level of work allocation if the number of new participants more than compensate for the dropout. 

In Fig.~\ref{fig:inflation}, we examine the increased total cost from the requester's side as $k$ decreases.  
\begin{figure}[!t]
    \centering
    \includegraphics[width=1\linewidth]{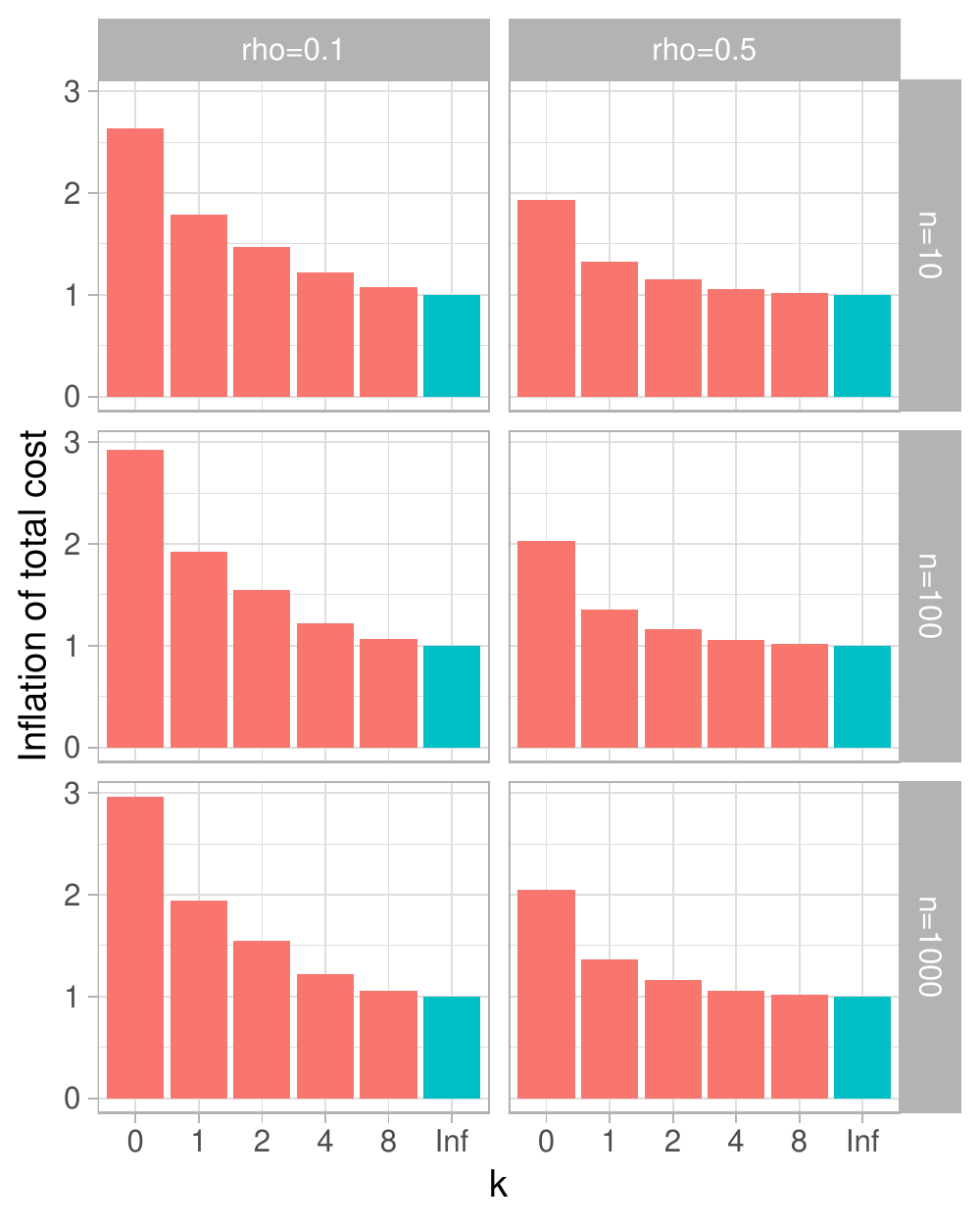}
    \caption{Relationship between inflation of total cost and $k$}
    \label{fig:inflation}
\end{figure}  
Theorem \ref{theorem4} shows that $k=\infty$ corresponds to the optimal allocation from the efficiency point of view. As $k$ decreases, the total cost increases as a result of the increased equality in allocation. The inflation is more apparent in the  case of $\rho=0.1$, and slightly more apparent in the case of $n=100$ or $n=1000$, over the case of $n=10$. 

Finally, we examine participation levels as a function of total costs in Fig.~\ref{fig:inflationROI}. 
\begin{figure}[!t]
    \centering
    \includegraphics[width=1\linewidth]{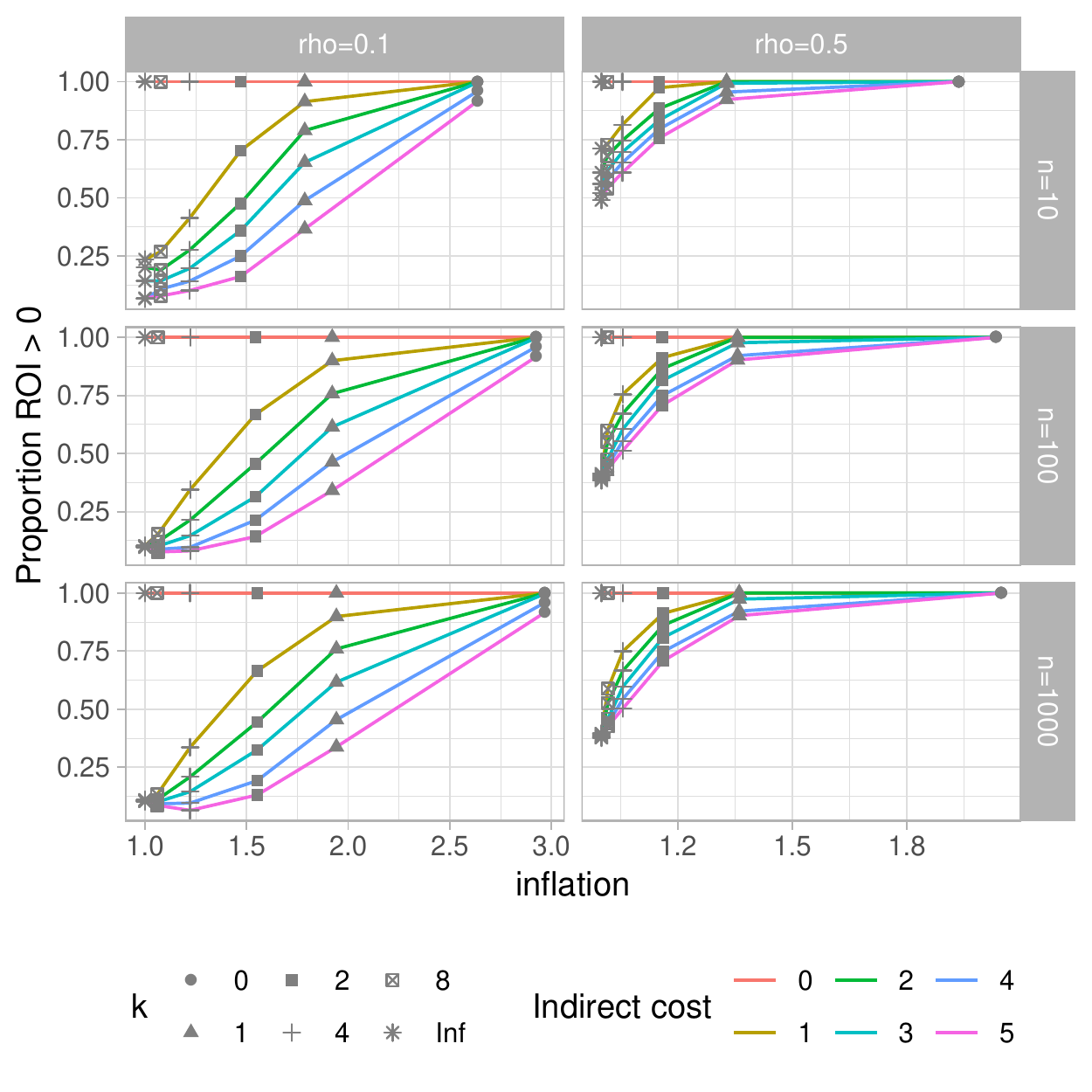}
    \caption{Relationship between inflation of total cost and proportion of $ROI$ greater than 1}
    \label{fig:inflationROI}
\end{figure}  
Since both expected long term participation and total cost are functions of $k$, it is evident that increasing participation must come at a cost of increasing total cost. We see that when $\rho=0.1$, to obtain significant increases in participation, one must at least double the minimum cost. However, when $\rho=0.5$, the cost of increasing participation is relatively less. Increasing the total spend, say, by 30\% can already achieve significant increases in participation levels. 

Overall, our simulations verify our claim that it is possible to control the expected return on investment ($ROI$) and therefore long-term participation by setting different values for the parameter $k$ (Figs.~\ref{fig:ROI} and \ref{fig:participation}). They also demonstrate that increasing participation rates can come at significant cost (Fig.~\ref{fig:inflation} and \ref{fig:inflationROI}), and the cost to achieve near 100\% retention of workers can be prohibitive to the requester. Apart from adjusting $k$ there are other strategies the requester could employ to retain workers. First, she can increase $\rho$, the ratio of the total available work ($c$) to the total number of workers ($n$). However, if $\rho$ approaches 1, then effectively all workers obtain work, and there is no longer competition leading to inflation of costs. Secondly, she can set $c$ independently of $n$. In this case, the cost of work will fluctuate according to the amount of competition in the auction. How $c$ and $k$ should be set in practice to maximize profit for a particular task needs to be determined on a case-by-case basis using actual data on the platform, and is beyond the scope of this paper. 

\section{Conclusion} \label{sec:conclusion}
In this paper, we extend the optimal auction theory for application on a crowdsourcing platform, where the bid for work consists not only the unit cost, but also the maximum amount of work the workers are willing to do and the actual work done. We prove that a dominant strategy exists in this case. This is especially valuable to the requester, since it guards against a worker performing only a small proportion of the allocated work, leading to inefficiency. We also prove that it is advantageous for the worker to complete as much of the allocated work as possible, and hence it is advantageous for him to truthfully report the maximum work he is able or willing to do in the bid. We also show that the dominant strategy is somewhat robust to different forms of the utility function that the worker may have, taking into account that there may be diminishing value of returns. 

Secondly, we propose and analyse a novel allocation mechanism, which allows the requester to balance between cost efficiency and equality in the allocation of work. Cost efficiency is obviously valuable to the requester, but the allocation of work only to the workers with the lowest bid may put off a large proportion of workers from future participation on the crowdsourcing platform. We show that increase in equality will lead to greater retention of workers through increase in their return on investment. Moreover, there is a tradeoff between cost and equality in allocation, and with sufficient prior knowledge on the distribution of bids, the requester can predict this tradeoff and select the best allocation that suits her needs. 

Finally, we note that in applying our crowdsourcing mechanism to a real online platform, there may be additional complexities due to the workers not being able to correctly predict the proportion of acceptable work that they supply. For example, a worker may not complete the work assigned because he overestimates his capacity when he first bids. In general, it is also possible that the workers are not aware of the dominant strategy and does not bid in the most advantageous manner. In principle, additional rewards or penalties (for failing to provide a certain level of acceptable work) may provide additional safeguards for truthful reporting of the workers' true valuation of the work. Future work to examine these measures theoretically as well as on a real crowdsourcing platform would be important next steps for research. 

\appendices

\section{Relaxing the assumptions on the form of the utility function} \label{sec:relax}
In \eqref{eq:utility00}, we assume that the utility is 
\begin{equation}
U_i = 
\begin{cases}
\tilde{p}_i - \hat{x}_i v_i & \text{if } \hat{x}_i \leq x^{\max}_i \\
0 & \text{if } \hat{x}_i > x^{\max}_i
\end{cases} \label{eq:utility000}
\end{equation}
The utility function may be considered unrealistic in the sense that it drops too suddenly once $\hat{x}_i$ becomes greater than $x^{\max}_i$. Here, we show that the form of the utility function can be relaxed and the conclusions of this paper will still hold. To see this, we first rewrite \eqref{eq:utility000} as 
\begin{align*}
U_i &= 
\omega(\hat{x}_i) (\tilde{p}_i - \hat{x}_i v_i) 
\end{align*} 
where 
\begin{align*}
\omega(\hat{x}_i) &= 
\begin{cases}
1 & \text{if } \hat{x}_i \leq x^{\max}_i \\
0 & \text{if } \hat{x}_i > x^{\max}_i
\end{cases}
\end{align*}
and $\omega(.)$ can be seen as a downweighting function. It may be observed that in the proof of Lemma \ref{Lemma2} (see Appendix \ref{sec:Lemma2}), the only place where the form of $U_i(.)$ when $\hat{x}_i > x^{\max}_i$ matters is in \eqref{eq:utility1} and \eqref{eq:utility2}. For simplicity, we do not consider the case where $\hat{x}_i > x_i$, since it is obviously not profitable for the worker to submit more work than the maximum requested. We rewrite \eqref{eq:utility1} and \eqref{eq:utility2} below using the downweighting notation
\begin{align}
U_i(\hat{x}_i) &= 
\begin{cases}
\omega(\hat{x}_i)\hat{x}_i\zeta_i & \text{if } \hat{x}_i > x^{\max}_i \\
\hat{x}_i\zeta_i & \text{if } \hat{x}_i \leq x_i  
\end{cases} \label{eq:utility1a} \\ 
\frac{dU_i}{d\hat{x}_i} &= 
\begin{cases}
\zeta_i(\omega_i  + \hat{x}_i \frac{d\omega_i}{d\hat{x}_i}) & \text{if } \hat{x}_i > x^{\max}_i \\
\zeta_i & \text{if } \hat{x}_i \leq x_i . \label{eq:utility2a}
\end{cases} 
\end{align}
where 
\begin{equation*}
\zeta_i = \frac{p_i\beta_i}{x_i} - v_i \geq 0, \quad \text{by Theorem \ref{theorem2}}
\end{equation*}

The same conclusion would be obtained as in Lemma \ref{Lemma2} if $\omega_i  + \hat{x}_i \frac{d\omega_i}{d\hat{x}_i} < 0$. Note that we require $\omega_i = 1$ for $\hat{x}_i \leq x_i$. Thus, we posit the following form for $\omega_i$:
\begin{equation*}
\omega_i(\hat{x}_i) = 1 -  \psi \max \left( \frac{\hat{x}_i - x^{\max}_i}{\hat{x}_i},0 \right)
\end{equation*}
such that if $\psi > 1$, then $\frac{d\omega_i}{d\hat{x}_i} < 0$. This implies that as long as $\omega(\hat{x}_i)$ satisfies, 
\begin{equation*}
\omega_i(\hat{x}_i)  
\begin{cases}
= 1 & \text{if } \hat{x}_i \leq x^{\max}_i \\
< 1 - \frac{\hat{x}_i - x^{\max}_i}{\hat{x}_i} = \frac{x^{\max}_i}{\hat{x}_i} & \text{if } \hat{x}_i > x^{\max}_i
\end{cases}, 
\end{equation*}
$U_i$ will continue to be maximized at $\hat{x}_i = x^{\max}_i$, and therefore our conclusions about the dominance of $(b_i, \hat{x}^{\max}_i, \hat{x}_i) = (v_i/\beta_i, x^{\max}_i, x_i)$ will continue to hold. 

\section{Proof of Lemma \ref{Lemma1}} \label{sec:Lemma1}
	The utility of Worker $i$ is given by 
	\begin{align}
	U_i(b_i) &= \E_{\alpha_i} \tilde{p}_i(b_i) - v_i\hat{x}_i(b_i) \nonumber \\
    &= \E_{\alpha_i} \alpha_i p_i(b_i) - v_ix_i(b_i) \nonumber \\
    &= \beta_i p_i(b_i) - v_ix_i(b_i) \nonumber \\
	&= \beta_i \left( b_i x_i(b_i) + \int_{b_i}^{\bar{b}_i} x_i(s) ds \right) - v_i x_i(b_i) \nonumber \\
	&= x_i(b_i)(b_i\beta_i - v_i) + \beta_i \int_{b_i}^{\bar{b}_i} x_i(s) ds \label{eq:utility0}
	\end{align}
	Now, let $b_i^*=v_i/\beta_i$, and $U_i^* = U_i(b_i^*) = \beta_i \int_{b_i^*}^{\bar{b}_i} x_i(s) ds$. Suppose $b_i > b_i^*$, we have
	\begin{align*}
	U_i(b_i) &= U_i^* - \beta\int_{b_i^*}^{b_i} x_i(s) ds - \beta(x_i(b_i^*)b_i^* - x_i(b_i)b_i) + \\
	&\qquad v_ix_i(b_i^*) - v_ix_i(b_i) \\
	&= U_i^* - \beta \left [ \int_{b_i^*}^{b_i} x_i(s) ds + (b_i - b_i^*)x_i(b_i) \right ] \\
	&\leq U_i^* 
	\end{align*}
	since $x_i(b_i)$ is non-increasing. Likewise, for $b_i < b_i^*$, 
	\begin{align*}
	U_i(b_i) &= U_i^* + \beta\int_{b_i}^{b_i^*} x_i(s) ds - \beta(x_i(b_i^*)b_i^* - x_i(b_i)b_i) + \\
	&\qquad v_ix_i(b_i^*) - v_ix_i(b_i) \\
	&= U_i^* + \beta \left [ \int_{b_i}^{b_i^*} x_i(s) ds - (b_i^* - b_i)x_i(b_i) \right ] \\
	&\leq U_i^* 
	\end{align*}
	since $x_i(b_i)$ is non-increasing. Hence $b_i = b_i^*$ is a dominant strategy. We also note that strategy is independent of the value of $\hat{x}^{\max}_i$. 

\section{Proof of Lemma \ref{Lemma2}} \label{sec:Lemma2}
	The utility of Worker $i$ is given by  
	\begin{equation}
	U_i(\hat{x}_i) = \begin{cases}
	0 & \text{if } \hat{x}_i > x^{\max}_i \\
	p_i\beta_i - \hat{x}_i v_i & \text{if } x_i \leq \hat{x}_i \leq x^{\max}_i \\
	p_i\frac{\beta_i \hat{x}_i}{x_i} - \hat{x}_i v_i & \text{if } \hat{x}_i \leq x_i  
	\end{cases} \label{eq:utility1}
	\end{equation}
	which implies 
	\begin{equation}
	\frac{dU_i}{d\hat{x}_i} = 
	\begin{cases}
	0 & \text{if } \hat{x}_i > x^{\max}_i \\
	-  v_i & \text{if } x_i < \hat{x}_i < x^{\max}_i \\
	\frac{p_i\beta_i}{x_i} - v_i & \text{if } \hat{x}_i < x_i . \label{eq:utility2}
	\end{cases} 
	\end{equation}
	Thus, $U_i$ is a piecewise linear function of $\hat{x}_i$. From \eqref{eq:utility0}, $U_i= \beta_ip_i(b_i) - v_ix_i(b_i) \geq 0$ if $b_i = v_i/\beta_i$. This implies $\frac{p_i\beta_i}{x_i} - v_i \geq 0$ in \eqref{eq:utility2}, and hence the maximum of \eqref{eq:utility1} is achieved at $\hat{x}_i = x_i$ if $x_i < x^{\max}_i$ or $x^{\max}_i$ otherwise. Thus, $\hat{x}_i=\min(x_i, x^{\max}_i)$ is a dominant strategy. 

\section{Proof of Lemma \ref{Lemma3}} \label{sec:Lemma3}
	In terms of $\hat{x}^{\max}_i$, and assuming $\hat{x}_i = \min(x_i, x^{\max}_i)$, we can write the utility function \eqref{eq:utility00} as
	\begin{equation}
	U_i(\hat{x}^{\max}_i) = P_i(\hat{x}^{\max}_i) p_i(\hat{x}^{\max}_i) - \min(x_i(\hat{x}^{\max}_i), x^{\max}_i)v_i \label{eq:utility3}
	\end{equation}
	where 
	\begin{align}
	P_i(\hat{x}^{\max}_i) &= \frac{\beta_i \min(x_i(\hat{x}^{\max}_i), x^{\max}_i)}{x_i(\hat{x}^{\max}_i)} \label{eq:P} \\
	p_i(\hat{x}^{\max}_i) &= x_i(\hat{x}^{\max}_i)b_i + \int_{b_i}^{\bar{b}_i} x_i(\hat{x}^{\max}_i, s) ds \label{eq:p}
	\end{align}
	Note that 
	\begin{equation}
	x_i(\hat{x}^{\max}_i) =
	\begin{cases}
	\hat{x}^{\max}_i & \text{if } \hat{x}^{\max}_i < x^{\crit}_i(b_i) \\
	x^{\crit}_i(b_i) & \text{if } \hat{x}^{\max}_i \geq x^{\crit}_i(b_i)
	\end{cases} \label{eq:xbehaviour}
	\end{equation}
	where $x^{\crit}_i(b_i) = \underset{\hat{x}^{\max}_i}{\max\ } x_i(\hat{x}^{\max}_i, b_i)$ is the critical value at which $\hat{x}^{\max}_i$ stops becoming limiting. 
	Thus, 
	\begin{equation}
	\frac{\partial x_i(\hat{x}^{\max}_i, b_i)}{\partial \hat{x}^{\max}_i} = 
	\begin{cases}
	1 & \text{if } \hat{x}^{\max}_i < x^{\crit}_i(b_i) \\
	0 & \text{if } \hat{x}^{\max}_i > x^{\crit}_i(b_i)
	\end{cases} \label{eq:diff}
	\end{equation}
	We will now prove that $U_i$ is maximized when $\hat{x}^{\max}_i=x^{\max}_i$ by considering three cases: 
	\begin{enumerate}
		\item $\hat{x}^{\max}_i \geq x^{\crit}_i(b_i)$,
		\item $\hat{x}^{\max}_i < \min(x^{\crit}_i(b_i), x^{\max}_i)$, and
		\item $x^{\max}_i < \hat{x}^{\max}_i < x^{\crit}_i(b_i)$.
	\end{enumerate}
	\subsection*{\textbf{Case 1}) $\hat{x}^{\max}_i \geq x^{\crit}_i(b_i)$} 
	By assumption, $x_i(b_i)$ is a non-increasing function of $b_i$. Hence, $x^{\crit}_i(b_i) = \underset{\hat{x}^{\max}_i}{\max\ } x_i(\hat{x}^{\max}_i, b_i)$ is a non-increasing function of $b_i$. Thus, if $\hat{x}^{\max}_i > x^{\crit}_i(b_i)$, then $\hat{x}^{\max}_i > x^{\crit}_i(s), \forall s > b_i$. This implies  
	\begin{equation}
	\frac{\partial }{\partial \hat{x}^{\max}_i} \int_{b_i}^{\bar{b}_i} x_i(\hat{x}^{\max}_i, s) ds = 0 \label{eq:diffint}
	\end{equation}
	and hence from \eqref{eq:p},\eqref{eq:diff}, and \eqref{eq:diffint}, we have
	\begin{equation*}
	p_i(\hat{x}^{\max}_i) = \text{constant}.
	\end{equation*} 
	From \eqref{eq:P},
	\begin{equation*}
	P_i(\hat{x}^{\max}_i) = 
	\begin{cases}
	\beta_i & \text{if } x^{\crit}_i(b_i) \leq x^{\max}_i \\
	\beta_i \frac{x^{\max}_i}{x^{\crit}_i(b_i)} < \beta_i & \text{if } x^{\crit}_i(b_i) > x^{\max}_i 
	\end{cases} 
	\end{equation*}
	and therefore from \eqref{eq:utility3}, 
	\begin{equation*}
	U_i(\hat{x}^{\max}_i) = 
	\begin{cases}
	\beta_i p_i - v_ix^{\crit}_i(b_i) & \text{if } x^{\crit}_i(b_i) \leq x^{\max}_i \\
	\frac{\beta_i x^{\max}_i}{x^{\crit}_i(b_i)}p_i - v_ix^{\max}_i & \text{if } x^{\crit}_i(b_i) > x^{\max}_i 
	\end{cases}
	\end{equation*}
	and does not depend on $\hat{x}^{\max}_i$. 
	
	\subsection*{\textbf{Case 2}) $\hat{x}^{\max}_i < \min(x^{\crit}_i(b_i), x^{\max}_i)$}
	From \eqref{eq:utility3}, \eqref{eq:P}, and \eqref{eq:xbehaviour}, 
	\begin{align*}
	U_i(\hat{x}^{\max}_i) &= \beta_i p_i(\hat{x}^{\max}_i) - \hat{x}^{\max}_iv_i \\ 
	\frac{\partial U_i}{\partial\hat{x}^{\max}_i} &= \beta_i \frac{\partial p_i}{\partial \hat{x}^{\max}_i}  - v_i 
	\end{align*} 
	and from \eqref{eq:p} and \eqref{eq:xbehaviour}, 
	\begin{align*}
	\frac{\partial p_i}{\partial \hat{x}^{\max}_i} &= b_i + \frac{\partial }{\partial \hat{x}^{\max}_i} \int_{b_i}^{\bar{b}_i} x_i(\hat{x}^{\max}_i, s) ds \geq b_i.
	\end{align*}
	Plugging in $b_i = v_i/\beta_i$, 
	\begin{equation}
	\frac{\partial U_i}{\partial\hat{x}^{\max}_i} = \beta_i\frac{\partial }{\partial \hat{x}^{\max}_i} \int_{v_i/\beta_i}^{\bar{b}_i} x_i(\hat{x}^{\max}_i, s) ds \geq 0, \label{eq:dU}
	\end{equation}
	where the presupposed condition is $\hat{x}^{\max}_i < \min(x^{\crit}_i, x^{\max}_i)$. Hence, in this case, $U_i$ is maximized when $\hat{x}^{\max}_i \rightarrow \min(x^{\crit}_i, x^{\max}_i)$. 
	
	\subsection*{\textbf{Case 3}) $x^{\max}_i < \hat{x}^{\max}_i < x^{\crit}_i(b_i)$}
	Plugging in $x_i(\hat{x}^{\max}_i)=\hat{x}^{\max}_i$ from \eqref{eq:xbehaviour} into \eqref{eq:P} and \eqref{eq:utility3}, we have 
	\begin{equation*}
	U_i(\hat{x}^{\max}_i) = \beta_i p_i(\hat{x}^{\max}_i)  \frac{x^{\max}_i}{\hat{x}^{\max}_i} - x^{\max}_i v_i 
	\end{equation*}
	For any given $(x^{\max}_i, \hat{x}^{\max}_i)$, define $k = x^{\max}_i/\hat{x}^{\max}_i$, and 
	\begin{align*}
	R(\hat{x}^{\max}_i, k) &= p_i(\hat{x}^{\max}_i)k \\
	R'(\hat{x}^{\max}_i, k) &= p_i(k\hat{x}^{\max}_i)
	\end{align*}
	such that 
	\begin{align}
	U_i(\hat{x}^{\max}_i) &= \beta_i R(\hat{x}^{\max}_i, k) - x^{\max}_iv_i \label{eq:util1} \\
	U_i(x^{\max}_i) &= \beta_i R'(\hat{x}^{\max}_i, k) - x^{\max}_iv_i.  \label{eq:util2}
	\end{align}
	By noting that $x_i(b_i)$ is non-increasing, and inverting the role of $b_i$ and $x_i$ in the integral of \eqref{eq:MyersonPay} and plugging in $x_i = \hat{x}^{\max}_i$, it can be observed that 
	\begin{equation*}
	p_i(\hat{x}^{\max}_i) = \int_0^{\hat{x}^{\max}_i} \min(\bar{b}, b_i(x_i)) dx_i
	\end{equation*}
	Writing 
	\begin{align*}
	r(\hat{x}^{\max}_i, s) &= p_i(\hat{x}^{\max}_i) \\
	r'(\hat{x}^{\max}_i, s) &= \hat{x}^{\max}_i\min(\bar{b}, b_i(s\hat{x}^{\max}_i)),
	\end{align*}
	we have 
	\begin{align*}
	R(\hat{x}^{\max}_i, k) &= \int_0^k r(\hat{x}^{\max}_i,s) ds \\
	R'(\hat{x}^{\max}_i, k) &= \int_0^k r'(\hat{x}^{\max}_i, s) ds
	\end{align*}
	Since $\min(\bar{b}, b_i(s\hat{x}^{\max}_i))$ is non-increasing, and $\int_0^1 r(\hat{x}^{\max}_i,s) ds = \int_0^1 r'(\hat{x}^{\max}_i, s) ds = p_i(\hat{x}^{\max}_i)$, $r(\hat{x}^{\max}_i, s)$ and $r'(\hat{x}^{\max}_i, s)$ must show the single-crossing property, i.e. there exists $t$ such that $r'(\hat{x}^{\max}_i, s) \geq r(\hat{x}^{\max}_i, s)$ for all $0 \leq s\leq t$ and $r'(\hat{x}^{\max}_i, s) \leq r(\hat{x}^{\max}_i, s)$ for all $t \leq s \leq 1$. Hence, for $0 \leq k \leq t$, $R'(\hat{x}^{\max}_i,k) \geq R(\hat{x}^{\max}_i, k)$ and for $t \leq k \leq 1$, $1 - R'(\hat{x}^{\max}_i, k) \leq 1 - R(\hat{x}^{\max}_i, k)$, implying 
	\begin{equation*}
	R'(\hat{x}^{\max}_i,k) \geq R(\hat{x}^{\max}_i, k), \quad \forall k.
	\end{equation*}
	Thus, with reference to \eqref{eq:util1} and \eqref{eq:util2}, we have $U_i(x^{\max}_i) \geq U_i(\hat{x}^{\max}_i)$ for all $(\hat{x}^{\max}_i, x^{\max}_i)$. 
	
	Gathering the results from \textbf{Case 2} and \textbf{Case 3}, if $x^{\max}_i < x^{\crit}_i(b_i)$, then $\hat{x}^{\max}_i=x^{\max}_i$ is a dominant strategy. If $x^{\max}_i \geq x^{\crit}_i(b_i)$ (\textbf{Case 1}), then any $\hat{x}^{\max}_i \geq x^{\crit}_i(b_i)$, including $\hat{x}^{\max}_i=x^{\max}_i$, is dominant. However, since $x^{\crit}_i(b_i)$ is unknown, $\hat{x}^{\max}_i=x^{\max}_i$ is the only dominant strategy. 

\section{Proof of Theorem \ref{theorem1}} \label{sec:theorem1}
    From Lemma \ref{Lemma2}, $\hat{x}_i = \min(x_i, x^{\max}_i)$ is a dominant strategy if $b_i = v_i/\beta_i$. However, if $b_i \neq v_i/\beta_i$, then the only other possibility is that $\frac{dU_i}{\hat{x}_i}=\frac{p_i\beta_i}{x_i} - v_i < 0$ in \eqref{eq:utility2}, resulting in $U_i$ being maximized at $\hat{x}_i = 0$. However, $\hat{x}_i = 0$ implies $U_i = 0$, and so it cannot be part of a global dominant strategy. Therefore a global dominant strategy, if it exists, must have $\hat{x}_i=\min(x_i, x^{\max}_i)$. 
    
    Now, given $\hat{x}_i=\min(x_i, x^{\max}_i)$, there are two possibilities: $\hat{x}_i = x_i \leq x^{\max}_i$ or $\hat{x}_i=x^{\max}_i < x_i$. In the former case, $\hat{x}_i=x_i\leq x^{\max}_i$ implies $b_i =v_i/\beta_i$ is a dominant strategy by Lemma \ref{Lemma1}. Together, by Lemma \ref{Lemma3}, we have $(b_i, \hat{x}^{\max}_i, \hat{x}_i)=(b_i, x^{\max}_i, x_i)$ as a possible global dominant strategy. However, we need to exclude the possibility that there exists a better strategy in $(b_i, \hat{x}^{\max}_i, \hat{x}_i)=(b_i \neq v_i/\beta_i, \hat{x}^{\max}_i \neq x^{\max}_i, x^{\max}_i)$. 
    
    To see that $(b_i, \hat{x}^{\max}_i, \hat{x}_i)=(b_i \neq v_i/\beta_i, \hat{x}^{\max}_i \neq x^{\max}_i, x^{\max}_i)$ is not a viable strategy, we note that in the proof of Lemma \ref{Lemma3}, the only place where $b_i \neq v_i/\beta_i$ may have an impact on the conclusion is in \eqref{eq:dU}, which would now become 
    \begin{equation*}
    \frac{\partial U_i}{\partial\hat{x}^{\max}_i} = \beta_i\frac{\partial }{\partial \hat{x}^{\max}_i} \int_{v_i/\beta_i}^{\bar{b}_i} x_i(\hat{x}^{\max}_i, s) ds + \beta_ib_i - v_i. 
    \end{equation*}
    We note that if $b_i > v_i/\beta_i$, then $\frac{\partial U_i}{\partial\hat{x}^{\max}_i} > 0$, and we still have the same conclusion that all other strategies would be dominated by $\hat{x}^{\max}_i = x^{\max}_i$. However, $\hat{x}^{\max}_i = x^{\max}_i$ implies $x_i \leq x^{\max}_i$, and hence $\hat{x}_i = x^{\max}_i < x_i$ is not possible. This leaves us with the possibility that $(b_i, \hat{x}^{\max}_i, \hat{x}_i)=(b_i < v_i/\beta_i, \hat{x}^{\max}_i \neq x^{\max}_i, x^{\max}_i)$ may present a better strategy. To see that this is not viable either, we reconsider the utility function in Lemma \ref{Lemma1}. Now, given $\hat{x}_i = x^{\max}_i$, we have
    \begin{align*}
    U_i(b_i) &= \beta_ix^{\max}_i\frac{p_i(b_i)}{x_i(b_i)} - x^{\max}_iv_i \\
    \frac{d \frac{p_i(b_i)}{x_i(b_i)}}{db_i} &= \frac{1}{x_i(b_i)^2}
    \left[
    x_i(b_i) \frac{d}{db_i} \int_{b_i}^{\bar{b}} x_i(s) ds - p_i(b_i)\frac{dx_i(b_i)}{db_i}
    \right]
    \end{align*}
    Since $\frac{d}{db_i} \int_{b_i}^{\bar{b}} x_i(s) ds = x_i(b_i)$ and $\frac{dx_i(b_i)}{db_i} \leq 0$, the term in the bracket is positive so long as $x_i(b_i)>0$. Hence, in the domain $0 <b_i \leq v_i/\beta_i$, $U_i$ is maximized at $b_i=v_i/\beta_i$. Thus, $b_i=v_i/\beta_i$ is the only viable dominant strategy, implying $(b_i, \hat{x}^{\max}_i, \hat{x}_i)=(b_i, x^{\max}_i, x_i)$ is a globally dominant strategy. 

\section{Proof of Lemma \ref{Lemma32}} \label{sec:Lemma32}
	First, we note that the constraints $x_i \geq 0$ are never tight if $\delta_i > 0$ for all $i$. This can be proved by a simple contradiction. Without loss of generality, suppose $x_1 = 0$ and $x_2 > 0$ is part of the optimal solution of \eqref{eq:problem}. To maintain the constraint $\sum_i x_i = c$, we let $x_i$ be increased by $\epsilon$ while $x_2$ be decreased by $\epsilon$. The resulting objective is $\delta_2^{k}x_2^2 + \epsilon(\epsilon(\delta_1^k+\delta_2^k) - 2\delta_2^kx_2) + \sum_{i > 2} \delta_i^k x_i^2$. Since $\epsilon(\epsilon(\delta_1^k+\delta_2^k) - 2\delta_2^kx_2)$ is a quadratic on $\epsilon$ with solutions at $\epsilon=0$ and $\epsilon(\delta_1^k+\delta_2^k) - 2\delta_2^kx_2=0$, there always exists $\epsilon>0$ such that $\delta_2^{2k}x_2^2 + \epsilon(\epsilon(\delta_1^k+\delta_2^k) - 2\delta_2^kx_2) <\delta_2^{2k}x_2^2$. Thus, $(x_1 = 0, x_2 > 0)$ cannot be optimal. 
	
	With the remaining constraints ($\sum_i x_i = c$, $x_i \leq \hat{x}^{\max}_i$), we apply the KKT conditions for the solution of a quadratic programming problem, and we have 

	\begin{equation}
	\begin{pmatrix}
	\text{diag}(\boldsymbol{\delta})^k & \boldsymbol{A} & -\boldsymbol{1} \\
	\boldsymbol{A}^T & \boldsymbol{0} & \boldsymbol{0} \\
	-\boldsymbol{1}^T & \boldsymbol{0}^T & 0 
	\end{pmatrix}
	\begin{pmatrix}
	\boldsymbol{x} \\
	\boldsymbol{\lambda}_{\mathcal{A}^*} \\
	\mu
	\end{pmatrix} 
	= 
	\begin{pmatrix}
	\boldsymbol{0} \\
	\hat{\boldsymbol{x}}^{\max}_i \\
	c
	\end{pmatrix}
	\label{eq:KKT0}
	\end{equation}
	where $\text{diag}(\boldsymbol{\delta})^k$ is the diagonal matrix whose diagonals are given by $(\delta_1, \ldots, \delta_n)$, $\boldsymbol{A}=(\boldsymbol{e}_1, \boldsymbol{e}_2, \ldots, \boldsymbol{e}_{|\mathcal{A}^*|})$, and $\boldsymbol{e}_j$ are vectors whose $i^{\text{th}}$ element is 1 and the rest are 0 if $i \in \mathcal{A}^*$. $\boldsymbol{\lambda}_{\mathcal{A}^*} > 0$ and $\mu$ are Lagrange multipliers. Solving \eqref{eq:KKT0} for $\boldsymbol{x}$ gives us the results. 

\section{Proof of Lemma \ref{Lemma3a}} \label{sec:Lemma3a}
        Let $q=r+1$. Denote by $a_i = \delta_i^{-k}$. Given $\hat{x}^{\max}_j, a_j > 0$ for all $j$, the following is implied by $\gamma_i < \mu(\mathcal{A}_r)$ and $q < i$: 
        \begin{align}
        \frac{\hat{x}^{\max}_i}{a_{i}} &< \mu(\mathcal{A}_r) = \frac{c - \sum_{j=1}^{r} \hat{x}^{\max}_j}{\sum_{j=q=r+1}^n a_j} \label{eq:cond1} \\
        \frac{\hat{x}^{\max}_q}{a_q} &\leq \frac{\hat{x}^{\max}_i}{a_i} \label{eq:cond2} 
        \end{align}
        Without loss of generality, let us assume $c - \sum_{j=1}^{r} \hat{x}^{\max}_j = c' =\sum_{j=q}^n a_j$, where $c' > \max{} (x^{\max}_q, a_q)$. (The generality of this rescaling is because Problem \eqref{eq:problem} is evidently scale invariant --- rescaling $\boldsymbol{b}$ as well as $\hat{\boldsymbol{x}}^{\max}$ and $c$ will lead to an equivalent allocation.) Then \eqref{eq:cond1} simplifies to 
        \begin{equation}
        \frac{\hat{x}^{\max}_i}{a_i} < 1 \label{eq:cond1a}
        \end{equation}
        We now prove by contradiction. Suppose contrary to expectation, we have $\gamma_i \geq \mu(\mathcal{A}_q)$. This implies 
        \begin{equation}
        \frac{\hat{x}^{\max}_i}{a_i} \geq \frac{c'-\hat{x}^{\max}_q}{c' - a_q}. \label{eq:cond4} 
        \end{equation}
        \eqref{eq:cond4} and \eqref{eq:cond2} imply 
        \begin{equation*}
        c'(\hat{x}^{\max}_i - a_i) \geq \hat{x}^{\max}_i a_q - \hat{x}^{\max}_q a_i \geq 0
        \end{equation*}
        which contradicts \eqref{eq:cond1a}. Therefore, $\gamma_i < \mu(\mathcal{A}_r) \implies \gamma_i < \mu(\mathcal{A}_{r+1})$. By induction, therefore, $\gamma_i < \mu(\mathcal{A}_r) \implies \gamma_i < \mu(\mathcal{A}_q)$ for all $r < q < i$.   

\section{Proof of Lemma \ref{Lemma_convergence}} \label{sec:Lemma_convergence}
	First, note that at iteration $m$, the $i$ whose $\gamma_i < \mu(\mathcal{A}^{(m-1)})$ will be added to $\mathcal{A}^{(m)}$. This parallels the tight constraint criterion in the KKT condition \eqref{eq:kkt2}, in which $\gamma_i < \mu(\mathcal{A}^*)$ are tight constraints. Thus, similar to the case of $\mathcal{A}^*$, we have $\mathcal{A}^{(m)} \in \mathbb{H}$ for all $m$. 
	
	Now it remains to show if $i \in \mathcal{A}^{(m)}$, then $i \in \mathcal{A}^*$. This is evident since $\mathcal{A}^{(0)}=\mathcal{A}_0$ at the first iteration. Now $i \in \mathcal{A}^{(1)} \text{ iff } \gamma_i < \mu(\mathcal{A}^{(0)})$. Hence all constraints that are in $\mathcal{A}^{(1)}$ will also be in $\mathcal{A}^*$, by Lemma \ref{Lemma3a}. Applying induction to $\mathcal{A}^{(m)}$ for $m > 1$, all tight constraints at the end of Algorithm \ref{algorithm} must be tight in $\mathcal{A}^*$. Since all of the constraints $\mathcal{A}^{(m)}$ are tight constraints in $\mathcal{A}^*$, the solution of Algorithm \ref{algorithm} must be a local minimum. Now because \eqref{eq:problem} is a quadratic programming problem with a strictly convex objective function, it must also be the global minimum. Conversely, $\mathcal{A}^{(m)}$ must correspond to $\mathcal{A}^*$ at the end of the algorithm. 

\section{Proof of Theorem \ref{theorem3}} \label{sec:theorem3}
	Since the solution of \eqref{eq:problem} can be obtained from Algorithm \ref{algorithm} by Lemma \ref{Lemma_convergence}, we have $x_i(b_i, \hat{x}^{\max}_i) = x^{(m)}_i = c'(\mathcal{A}^{(m-1)}) \delta_i^{-k}/\sum_{j \in {\mathcal{A}^{(m-1)}}^-} \delta_j^{-k}$ for some iteration $m$ in Algorithm \ref{algorithm} and $x_i(b_i, \hat{x}^{\max}_i)$ is a continuous function of $b_i$. It can be observed that 
	\begin{equation*}
	\frac{\partial x^{(m)}_i}{\partial \delta_i} < 0. 
	\end{equation*}
	Thus, if $\delta_i(b_i)$ is monotonically increasing in $b_i$, then as $b_i$ decreases, $x^{(m)}_i$ always increases, and therefore $x_i(b_i, \hat{x}^{\max}_i)$ must either increase or stay the same. 

\section{Proof of Theorem \ref{theorem4}} \label{sec:theorem4}
	According to \cite{Myerson1981,Koutsopoulos2013}, if an algorithm solves the optimization problem $\{\min\ \sum_i \delta_i x_i, \text{ s.t. } \sum_i x_i = c, 0 \leq x_i \leq \hat{x}^{\max}_i\}$, it minimizes expected costs. It is obvious that this is equivalent to allocating $c$ to the workers with the lowest values of $\delta_i$, subject to the constraints. Thus, if Algorithm \ref{algorithm} also performs this allocation as $k \rightarrow \infty$, it minimizes total expected costs. 
	
	To see this, we note that
	\begin{equation}
	\frac{\delta_i^k}{\delta_j^k} \underset{k \rightarrow \infty}{\rightarrow} 
	\begin{cases}
	\infty & \text{if } \delta_i > \delta_j \\
	1 & \text{if } \delta_i = \delta_j \\
	0 & \text{if } \delta_i < \delta_j
	\end{cases}. 
	\end{equation}
	At the $m^{\text{th}}$ iteration, as $k \rightarrow \infty$, 
	\begin{align}
	x_i^{(m)} &= 
	\begin{cases}
	\hat{x}^{\max}_i & \text{if } i \in \mathcal{A}^{(m-1)} \\
	\tilde{x}_i^{(m)} = \frac{c'({\mathcal{A}^{(m-1)}}^-)\delta_i^{-k}}{\sum_{j \in \mathcal{A}^{(m-1)}} \delta_j^{-k}}  & \text{otherwise} 
	\end{cases} \\
	&= 
	\begin{cases}
	\hat{x}^{\max}_i & \text{if } i \in \mathcal{A}^{(m-1)} \\ 
	\frac{c'(\mathcal{A}^{(m-1)})}{| \{\tilde{x}_j^{(m)} = \underset{q \in {\mathcal{A}^{(m-1)}}^-}{\min} \tilde{x}_q^{(m)} \} |} & \text{if } \tilde{x}_i^{(m)} = \underset{q \in {\mathcal{A}^{(m-1)}}^-}{\min} \tilde{x}_q^{(m)} \\
	0 & \text{otherwise} 
	\end{cases}
	\end{align}
	If $c'(\mathcal{A}^{(m-1)})/| \{\tilde{x}_j^{(m)} = \underset{q \in {\mathcal{A}^{(m-1)}}^-}{\min} \tilde{x}_q^{(m)} \} | > \hat{x}^{\max}_i$, $i$ is added to $\mathcal{A}^{(m)}$. Therefore, only the $i$ with the lowest $\tilde{x}_i$ will be added to the $\mathcal{A}^{(m)}$. Hence, we are allocating all work to workers with the lowest negative virtual welfare. 

\section{Proof of Lemma \ref{LemmaSC}} \label{sec:LemmaSC}
	In the following, we shall write $\pi^{(k, \boldsymbol{b})}(b_i)$ equivalently as $\pi^{(k, \boldsymbol{\delta})}(\delta_i)$, since $\pi(.)$ is dependent on $\boldsymbol{b}$ through $\boldsymbol{\delta}$. First note that if we linearly join up all the points corresponding to $\{(\delta_i, \pi^{(k, \boldsymbol{\delta})}(\delta_i)): i = 1,\ldots, n\}$ and $\{ (\delta_i, \pi^{(k+1, \boldsymbol{\delta})}(\delta_i)): i=1,\ldots,n\}$ along $\delta_i$, the two lines must either overlap or cross over at least once. If they do not, then $\sum_i \pi^{(k)} = \sum_i \pi^{(k+1)} = 1$ cannot hold. To show that $\pi^{(k, \boldsymbol{\delta})}(\delta_i)$ and $\pi^{(k+1, \boldsymbol{\delta})}(\delta_i)$ cross over \emph{only} once, note that 
	\begin{align}
	x^{(k)}_i(\delta_i) &= 
	\begin{cases}
	\hat{x}^{\max}_i & \text{if } i \in \mathcal{A}^{(k)} \\
	c'(k) \frac{\delta_i^{-k}}{\sum_{j \in {\mathcal{A}^{(k)}}^-} \delta_j^{-k}} & \text{if } i \in {\mathcal{A}^{(k)}}^- 
	\end{cases} \\
	&= \min \left(\hat{x}^{\max}_i, c'(k) \frac{\delta_i^{-k}}{\sum_{j \in {\mathcal{A}^{(k)}}^-} \delta_j^{-k}} \right). \label{eq:min}
	\end{align}
	where \eqref{eq:min} is by virtue of the fact that $x_i = \hat{x}^{\max}_i \implies \lambda_i = c'(k)(\sum_{j \in \mathcal{A}^-} \delta_i^{-k})^{-1} - x_i \delta_i^k \geq 0 \implies \hat{x}^{\max}_i \leq c'(k)\delta_i^{-k}(\sum_{j \in \mathcal{A}^-} \delta_i^{-k})^{-1}$. Since $g_1(\delta) < g_2(\delta) \implies \min(g_1(\delta), \alpha) \leq \min(g_2(\delta), \alpha)$ for arbitrary functions $g_1, g_2$ and constant $\alpha$, the crossing point(s) of $\pi^{(k, \boldsymbol{\delta})}(\delta_i)$ and $\pi^{(k+1, \boldsymbol{\delta})}(\delta_i)$ must be determined solely by $\{(x^{(k)}_i, x^{(k+1)}_i) : i \in {\mathcal{A}^{(k)}}^- \cap {\mathcal{A}^{(k+1)}}^- \}$. Now note that $x^{(k)}_i(\delta_i) \propto \delta_i^{-k}$ if $i \in {\mathcal{A}^{(k)}}^-$. Since the two functions $c'(k)\delta^{-k}$ and $c'(k+1)\delta^{-(k+1)}$ have only one crossing over point for $\delta > 0$, $\pi^{(k)}(\delta_i)$ and $\pi^{(k+1)}(\delta_i)$ must cross over \emph{only} once. 

\section{Proof of Theorem \ref{theorem5}} \label{sec:theorem5}
    From \cite{Myerson1981,Koutsopoulos2013}, we have 
    \begin{equation}
    \E_{\boldsymbol{b}} \sum_i p_i(b_i) = \E_{\boldsymbol{\delta}} \sum_i x_i^{(k)}(\boldsymbol{\delta}) \delta_i.
    \end{equation}
    where $(k)$ reminds us that $x_i$ is dependent on the tuning parameter $k$. Since we are optimizing for $\boldsymbol{x}$ subject to $\sum_i x_i = c$, we can write 
    \begin{align}
    \E_{\boldsymbol{\delta}} \sum_i x^{(k)}_i(\boldsymbol{\delta}) \delta_i &= c \E_{\boldsymbol{\delta}} \sum_i \pi^{(k, \boldsymbol{\delta})}(\delta_i) \delta_i
    =  c\E_{\boldsymbol{\delta}} \E_{\pi^{(k, \boldsymbol{\delta})}} (\delta) \nonumber \\
    &= c\E_{\boldsymbol{b}} \E_{\pi^{(k, \boldsymbol{b})}} \delta(b) \label{eq:totalcost}
    \end{align}
    From Lemma \ref{LemmaSC}, we have $\pi^{(k, \boldsymbol{b})}(b_i) \leq \pi^{(k+1, \boldsymbol{b})}(b_i)$ for all $b_i < t$ and $\pi^{(k, \boldsymbol{b})}(b_i) \geq \pi^{(k+1, \boldsymbol{b})}(b_i)$ for all $b_i > t$, for some $t$. Defining 
    \begin{equation*}
    	G^{(k, \boldsymbol{b})}(b) = \sum_{j: b_j \leq b} \pi^{(k, \boldsymbol{b})}(b_j), 
    \end{equation*}
    it follows that 
    \begin{align}
    	G^{(k, \boldsymbol{b})}(b_i) &\leq G^{(k+1, \boldsymbol{b})}(b_i), \quad b_i < t \nonumber \\
    	1 - G^{(k, \boldsymbol{b})}(b_i) &\geq 1 - G^{(k+1, \boldsymbol{b})}(b_i), \quad b_i > t \nonumber \\
    	\implies G^{(k, \boldsymbol{b})}(b_i) & \leq G^{(k+1, \boldsymbol{b})}(b_i), \quad \forall b_i. \label{eq:G}
    \end{align}
    Now, 
    \begin{align*}
		\E_{\pi^{(k, \boldsymbol{b})}} & \delta(b) - \E_{\pi^{(k+1, \boldsymbol{b})}} \delta(b) \\
		&= \int_0^{\bar{b}} \delta(b) dG^{(k, \boldsymbol{b})}(b) - \int_0^{\bar{b}} \delta(b) dG^{(k+1, \boldsymbol{b})}(b) \\
		&= \delta(b)[G^{(k, \boldsymbol{b})}(b) - G^{(k+1, \boldsymbol{b})}(b)]_0^{\bar{b}} - \\
		&\quad \int_0^{\bar{b}} \delta'(b) G^{(k, \boldsymbol{b})}(b) db + \int_0^{\bar{b}} \delta'(b) G^{(k+1, \boldsymbol{b})}(b) db \\
		&= \int_0^{\bar{b}} \delta'(b) G^{(k+1, \boldsymbol{b})}(b) db - \int_0^{\bar{b}} \delta'(b) G^{(k, \boldsymbol{b})}(b) db \\
		& \geq 0 \quad \text{by \eqref{eq:G}}
    \end{align*} 
    Hence, by \eqref{eq:totalcost}, the total expected cost is a non-increasing function of $k$. 

\bibliographystyle{IEEEtran}
\bibliography{MyCollection}

\end{document}